\documentclass[prodmode,acmec]{ec-acmsmall}

\acmVolume{X}
\acmNumber{X}
\acmArticle{X}
\acmYear{2013}
\acmMonth{2}
\usepackage[numbers]{natbib}
\usepackage{amsmath}
\usepackage{amsfonts, amssymb}
\usepackage{algorithm}
\usepackage[noend]{algpseudocode}
\usepackage{subfigure}
\usepackage{color}
\usepackage{multirow}

\usepackage{tikz}
\usetikzlibrary{shadows,arrows,decorations,decorations.shapes,backgrounds,shapes}

\usepackage[font=small,format=plain,labelfont=bf,up]{caption}
\newcount\Comments
\definecolor{darkgreen}{rgb}{0,0.7,0}
\newcommand{\kibitz}[2]{\ifnum\Comments=1\textcolor{#1}{#2}\fi}
\newcommand{\ap}[1]{\kibitz{blue} {[AP: #1]}}
\newcommand{\os}[1]{\kibitz{darkgreen} {[OS: #1]}}
\newcommand{\jb}[1]{\kibitz{red} {[JB: #1]}}
\newcommand{\sk}[1]{\kibitz{orange} {[SK: #1]}}

\newtheorem{claim}[theorem]{Claim}

\newcommand{\cut}[1]{}

\newcommand{\PasswordSpace}{\mathcal{P}}
\newcommand{\poly}{\ensuremath{\mathrm{poly}}}
\providecommand{\e}[1]{\ensuremath{\times 10^{#1}}}

\def \QED {\hfill{$\Box$}}

\newenvironment{proofof}[1]{\noindent {\em Proof of #1.  }}{\QED}

\newenvironment{remindertheorem}[1]{\medskip \noindent {\bf Reminder of Theorem #1.  }\em}{}

\newenvironment{reminderclaim}[1]{\medskip \noindent {\bf Reminder of Claim #1.  }\em}{}

\newcommand{\TheTitle}{Optimizing Password Composition Policies}

\begin{document}

\markboth{J. Blocki et al.}{\TheTitle}

\title{\TheTitle} 


\author{
JEREMIAH BLOCKI  \affil{Carnegie Mellon University} 
SARANGA KOMANDURI \affil{Carnegie Mellon University}
ARIEL D. PROCACCIA  \affil{Carnegie Mellon University} 
OR SHEFFET\affil{Carnegie Mellon University}
}

\begin{abstract}
A password composition policy restricts the space of allowable passwords to eliminate weak passwords that are vulnerable to statistical guessing attacks. 
Usability studies have demonstrated that existing password composition policies can sometimes result in weaker password distributions; hence a more \emph{principled} approach is needed. We introduce the \emph{first theoretical model} for optimizing password composition policies. We study the computational and sample complexity of this problem under different assumptions on the structure of policies and on users' preferences over passwords. Our main positive result is an algorithm that -- with high probability --- constructs almost optimal policies (which are specified as a union of subsets of allowed passwords), and requires only a small number of samples of users' preferred passwords. We complement our theoretical results with simulations using a real-world dataset of $32$ million passwords. 
\end{abstract}


\terms{Algorithms, Economics, Security, Theory}


\acmformat{Jeremiah Blocki, Saranga Komanduri, Ariel D. Procaccia
and Or Sheffet, 2013. \TheTitle.}

\begin{bottomstuff}
Authors' addresses: J. Blocki, Computer Science Department,
Carnegie Mellon University, email: \texttt{jblocki@cs.cmu.edu}; S. Komanduri, Human Computer Interaction Institute, Carnegie Mellon University, email: \texttt{sarangak@cs.cmu.edu}; A. D. Procaccia, Computer Science Department, Carnegie Mellon University, email: \texttt{arielpro@cs.cmu.edu}; O. Sheffet, Computer Science Department, Carnegie Mellon University, email: \texttt{osheffet@cs.cmu.edu}. \\
This research was supported in part by the National Science Foundation Science and Technology TRUST, by the National Science Foundation under grants DGE-0903659, CNS-1116776, CCF-1101215 and CCF-1116892, by CyLab at Carnegie Mellon under grants DAAD19-02-1-0389 and
W911NF-09-1-0273 from the Army Research Office, by the AFOSR MURI on Science of Cybersecurity, by a gift from Microsoft Research and by a NSF Graduate Research Fellowship.

\end{bottomstuff}

\maketitle
\section{Introduction} \label{sec:Introduction}
\def \newintro {}
\ifdefined\newintro
{
Imagine a web surfer, an online shopper, or a reviewer in a prominent CS and Economics conference\footnote{All three might be the same person.} who logs on for the first time to a server; so that she can sign up for some service,  place a shopping order, or view a list of assigned papers. Such a user registers on the server by choosing a username and picking a password. Naturally, our user's first attempt at picking a password is her favorite combination `\verb"123456"', which the server declines. She then has to pick a password that follows certain guidelines: of suitable length, involving lower- and upper-case letters, with numbers or special characters, etc. Such \emph{password composition policies} defend against the ``first line'' of attack~-- guessing attacks by uninformed attackers (attackers with no previous knowledge of the user whose account they are trying to break into). 

Password composition policies are a necessity because --- without them --- user-selected passwords are predictable. Indeed, many unrestricted users would select simple passwords like `\verb"123456"', `\verb"password"' and `\verb"letmein"' \cite{mostPopularPasswords2012}. Furthermore, this issue is of great importance to today's economy. Passwords are commonly used in electronic commerce to protect financial assets. In fact, the passwords themselves have financial value. Symantec reported that compromised passwords are sold for between \$4 and \$30 on the black market \cite{passwordBlackMarket}, and a 2004 Gartner case study \cite{costOfPasswordReset} estimated that it cost a large firm over \$17 per password-reset call. 
Nevertheless, existing password composition policies are typically not principled, and do not necessarily result in less common passwords. For example, studies show that users respond to restrictions in predictable ways~ \cite{komanduri2011passwords}, or pick weaker passwords due to user-fatigue~\cite{clair2006password,kruger2008empirical}.

In this paper, we initiate the \emph{algorithmic} study of password composition policies. Such policies restrict the space of passwords to a subset of allowed passwords, and force each user to pick a password in this subset. Thus, $n$ users induce a distribution over passwords where for a password $w$, $\Pr[w] = \tfrac 1 n \left|\left\{i~:~ i\textrm{ picks } w \right\}\right|$. By declaring different subsets of allowed passwords, different password composition policies induce different distributions. Our work formalizes and addresses the algorithmic problem a server administrator faces when designing a password composition policy; we ask:

\begin{quote}
\emph{In what settings can the information about the users' preferences over passwords allow us to design a password composition policy that is \emph{guaranteed} to induce a password distribution as close to uniform as possible?}
\end{quote}

We wish to stress at this point that we do not take a cryptographic approach to the problem: we do not design a protocol aimed at amplifying a password's strength, nor do we rely on standard cryptographic assumptions or techniques in designing our password composition policies. Single-factor authentication does not defend against an attacker who learns about the most probable password from an external source. Furthermore, because password systems often allow users multiple attempts in entering their password, an attacker can make a small number of guesses with impunity. Therefore, we instead focus on the design and analysis of algorithms for optimizing the password composition policy's induced distribution over passwords, and in our theoretical results compare the performance of our algorithm to the optimal policy among exponentially many potential policies in the \emph{worst case}. 

\subsection{Our Model}

We study the algorithmic problem of optimizing password composition policies along multiple dimensions: the goal, the user model, and the policy structure.

\smallskip
\noindent\textbf{Goal.} We focus on designing a policy that maximizes the minimum-entropy of the resulting password distribution. Specifically, we assume the server deals with $n$ users, each picking a password from some space of passwords $\PasswordSpace$ that respects the server's password composition policy. These $n$ passwords form a distribution over the domain of all allowed passwords and our goal is to minimize the probability of the most likely password. This is a natural goal (see Section~\ref{sec:disc}), as opposed to maximizing the Shannon-entropy of the distribution, which for example is still high even if half the people choose the same password and the other half choose a password uniformly at random from $\PasswordSpace$. From a security standpoint, the minimum entropy represents the fraction of accounts that could be compromised in one guess. For example, an adversary would be able to crack 
$0.9\%$ of RockYou passwords \cite{rockYouPasswords} with only one guess. 
Alternatively, should the attacker attempt to break into only one account, the minimum entropy represents the likelihood that the account is compromised on the first guess.
We also consider a slightly stronger goal of minimizing the fraction of accounts that could be compromised using $k$ guesses, that is, the overall probability of the $k$ most likely passwords \cite{boztas1999entropies}.

\smallskip
\noindent\textbf{User model.} 
We consider two models for how users select passwords when presented with a password composition policy. 

In the \emph{ranking model}, each user has an implicit ranking over passwords, from the most preferred to the least preferred. Given a password policy, each user selects the highest-ranking password among those allowed by the policy. There is a distribution over the space of rankings that determines the fraction of users with each possible ranking. Note that for any password composition policy, such a distribution over rankings induces a distribution over the most preferred \emph{allowed} passwords. 

In the \emph{normalization model}, there is a distribution $\mathcal{D}$ over the space of all passwords. This distribution tells us the likelihood that an unrestricted user would select a given password. Given a password composition policy, $\mathcal{D}$ induces a new distribution over the allowed passwords (which can be obtained by normalizing the probabilities under $\mathcal{D}$ of the allowed passwords). When we ban a password the fraction of users that prefer each allowed password grows; the natural interpretation is that users who preferred an allowed password still use that password, but users who preferred a banned password are redistributed among the allowed passwords according to the induced distribution. 
 
As we show, the normalization model is strictly more restrictive than the ranking model: any distribution in the normalization model can be simulated in the ranking model, but there exist hardness results for the ranking model that do not hold for the normalization model.



\smallskip
\noindent\textbf{Policy structure.} 
We consider the best policy that is restricted to manipulation of a given set of \emph{rules} --- each rule is simply a predefined subset of potential passwords. These rules are given to us as part of the problem (see Section~\ref{sec:disc} for a discussion of this point). 
If we interpret a rule as a subset of banned passwords (e.g., passwords shorter than seven characters), its complement (e.g., passwords of at least seven characters) can be interpreted as a subset of allowed passwords. As such, when we take the union of rules we get either a set of banned passwords (\emph{negative rules}) or allowed passwords (\emph{positive rules}); this is our password composition policy. While the distinction between the two cases may at first seem a mere technicality, it is in fact quite significant due to the following observation. If we ban the union of rules then in order to ban a password that was picked by too many users, we may ban any rule that contains this password. In contrast, if we allow a union of rules then in order to ban this password we must not allow \emph{any} rule that contains it. In other words, when our goal is to discard a password in the negative rules setting, we have multiple ways to do so. When our goal is to discard a password in the positive rules setting, we have only one way to do so --- excluding all rules that allow this password. As we shall see, this seemingly small difference leads to a clear separation between the two scenarios in terms of the complexity of designing optimal policies. 

We pay special attention to the case where each password has its own singleton rule. In this setting, a policy can be interpreted as a ``blacklist'' of banned passwords that do not necessarily share common characteristics. Note that when each password has its own singleton rule, it does not matter whether these rules are positive or negative.

\subsection{Our Results}
\label{subsec:results}

As we noted above, a password composition policy induces a distribution over most preferred passwords (in both user models). Hence we can study algorithms that \emph{sample} these distributions. One can obtain such samples by asking random users to choose a password that is constrained by a certain policy. Clearly, though, we need the number of samples to be ``small''. The size of the space of all passwords $\PasswordSpace$ --- which we denote by $N$ --- is typically very large (e.g., $\PasswordSpace$ can include all passwords that are no longer than $32$ ASCII characters). We wish to maximize entropy using a number of samples that does not depend on $N$. 

Before tackling this goal directly, we study the problem in a simpler setting where the preferences of all users are given to us as input (i.e., there is no uncertainty). In particular, here $\PasswordSpace$ is a part of the input and algorithms are allowed to run in time polynomial in $N$. The computational complexity of problems in this setting informs their study in the sampling setting: it is hopeless to design efficient sampling algorithms for problems that are computationally hard, but computationally tractable problems may (or may not) have efficient sampling algorithms. 

Table~\ref{tab:Usability} summarizes our complexity results. The parameter $k$ refers to our optimization target: minimizing the likelihood of the $k$ most likely passwords. \os{Added the following sentence.} Some results are direct corollaries of others --- using the fact that singleton rules are a special case of positive rules and the fact that the normalization model is a special case of the ranking model (see Section~\ref{sec:Models}). Looking at the table one immediately notices a clear separation between negative rules and positive rules: optimization using the latter is much easier. 

\renewcommand{\arraystretch}{1.3}
\renewcommand{\tabcolsep}{4pt}
\begin{table}[t]
\centering
 \tbl{Summary of Complexity Results.}{
\begin{tabular}{| l | p{1in} | p{1in}  ||  p{1in} | p{1in} | }
\cline{2-5}
\multicolumn{1}{c|}{}& \multicolumn{2}{c||}{\textbf{Ranking Model}} & \multicolumn{2}{c|}{\textbf{Normalization Model}} \\
\cline{2-5}
\multicolumn{1}{c|}{}& \textbf{Constant $k$}  & \textbf{Large $k$}  & \textbf{Constant $k$} & \textbf{Large $k$} \\
\hline
\textbf{Singleton rules} & P & NP-Hard (Thm \ref{thm:hardnessKisParameter})\newline APX-Hard w/ UGC (Thm~\ref{thm:UGCHardness}) & P & P (Thm \ref{thm:NrmPrbEfficientAlgAnyK}) \\
\hline
\textbf{Positive rules} & P (Thm \ref{thm:ConstantKAlg} ) & NP-Hard & P & NP-Hard (Thm~\ref{thm:HardnessAddingSubsetsLargeK})  \\
\hline 
\textbf{Negative rules} & $n^{1/3}$-approx is NP-hard (Thm~\ref{thm:HardnessOfApproximationOptimalCombinationOfSubsets}) & NP-Hard  & NP-Hard (Thm \ref{thm:NrmPrbHardnessOptimizeP1}) & NP-Hard \\
\hline
\end{tabular}}
\label{tab:Usability}
\end{table}

We therefore focus on positive rules in our attempt to design an efficient sampling algorithm. Our main result is the best one could hope for in this setting. We design an algorithm that works in the more general ranking model, and finds a policy whose entropy is $\epsilon$-close to optimal with probability $1-\delta$, for any given $\epsilon,\delta>0$. The required number of samples is polynomial in $1/\epsilon$, $\log(1/\delta)$, and the number of positive rules $m$. We can assume that $m$ is small, because each rule corresponds to a subset of passwords that can be concisely described to users.

These results can be applied in a practical setting, and we show this through simulated sampling experiments using natural rules and a large dataset of real passwords. The experimental results provide evidence for the difficulty of the negative rules setting: we search all combinations of rules to find the optimal policy and then attempt to discover this policy by making decisions both randomly and with a heuristic. In the negative rules setting, neither approach succeeded at finding the optimal policy after hundreds of iterations at various sample sizes, and average-case performance did not improve with sample size. In the positive rules setting, the average-case performance of our efficient algorithm improved with sample size and, with a moderate sample size, found 
policies that were either optimal or very close to optimal. 

}
\else
\input{intro_Feb8th}
\fi

\subsection{Related Work}


It has been repeatedly demonstrated that users tend to select easily guessable passwords \cite{rockYouPasswords,mostPopularPasswords2012,bonneau2012science} and NIST recommends that organizations ``should also ensure that other trivial passwords cannot be set," to thwart potential attackers \cite{NIST-Passwords}. Unfortunately, this task is more difficult than it might appear at first. Policies were initially developed without empirical data to support them, since such data was not available to policy designers~\cite{burr_electronic_2006}. When hackers leaked the RockYou dataset to the Internet, both researchers (and attackers) suddenly had access to password data, leading to many insights into true passwords~\cite{weir_testing_2010}. However, recent research analyzing leaked datasets from non-English speakers, notably Hebrew and Chinese-language websites, shows that trivial password choices can vary between contexts, making a simple blacklist approach ineffective~\cite{bonneau_contrasenas_2012}. This means that, depending on the context, a policy based on leaked password data might provide no security guarantee, and it has ethical issues as well.

To combat this issue, researchers have turned to a sampling approach. Bonneau~\citeyear{bonneau2012science} added a system for sampling to the Yahoo! password infrastructure. This system allows one to gain empirical data about the frequency distribution of passwords without revealing the passwords themselves. Such approaches provide a way of gathering empirical data about passwords while maintaining the anonymity of users. Our algorithms could be used in conjunction with such an infrastructure to optimize policies.

Komanduri et al.~\citeyear{komanduri2011passwords} studied the effectiveness of several basic password composition policies by using Amazon's Mechanical Turk to conduct a large scale user study. They found that people often respond to restrictions in predictable ways (e.g., if the password needs to contain a capital letter users might tend to capitalize the first letter of a password) and provide very general recommendations for password composition policies. However, no theoretical model has been proposed for studying the password composition problem.

Schechter et al.~\citeyear{schechter2010popularity} suggest using a popularity oracle to prevent individual passwords that have been used too frequently from being selected by new users. They also proposed using the count-min sketch data structure \cite{cormode2005improved} to build such a popularity oracle. Malone and Maher~\citeyear{malone2012investigating} suggest a similar system using a Metropolis-Hastings scheme to force an approximately uniform distribution on passwords. Usability results on the effectiveness of dictionary checks~\cite{komanduri2011passwords} suggest that such policies would be very frustrating since the policy is hidden from users behind an oracle. In contrast, we seek to construct optimal policies from combinations of rules that are visible to the user and can be described in natural language. 

This consideration of users is important to electronic commerce, even where security is concerned. Florencio and Herley~\citeyear{florencio2010security} studied the economic factors that drive institutions to adopt strict password composition policies and find that they often value the user experience over security. 
An e-mail provider like Yahoo! might adopt simple composition policies because a frustrated user could easily switch to Gmail, while universities are free to adopt strict policies because users cannot switch easily.

\section{A Model of Password Composition Policies} \label{sec:Models}

We use $\PasswordSpace$ to denote the space of all possible passwords. $N = \left| \PasswordSpace\right|$ is used to denote the total number of passwords. We denote the number of users by $n$.

A password composition policy may be specified in terms of rules. A rule is a subset of passwords $R \subseteq \PasswordSpace$ (e.g., the set of all passwords with more than seven characters). We use $R_1,...,R_m$ to denote a list of rules that may be active or inactive. We consider two schemes. 

\begin{itemize}
\item {\em Positive Rules: } A password $w$ is allowed if and only if it is allowed by some active positive rule. Formally, a password composition policy  $\mathcal{A}_S = \bigcup_{i \in S} R_i$
is specified by a set $S \subseteq [m] = \{1,...,m\}$ of active rules. In this setting rules should consist of sets of passwords which we expect to be strong (e.g., $R_i$ might be the set of all passwords longer than 10 characters, or the set of all passwords that use both upper and lowercase letters, or the set of all passwords that do not include a dictionary word). 

\item {\em Negative Rules:} A password $w$ is allowed if and only if it is not contained in any active negative rule. Formally, a solution  $\mathcal{A}_S = \left\{w \in \PasswordSpace ~\vline~w \notin \bigcup_{i \in S} R_i \right\}$ is given by a subset $S \subseteq [m]$ of active rules. A negative rule should consist of passwords that we expect to be weak (e.g., $R_i$ might be the set of all passwords without an uppercase letter, or the set of all passwords shorter than 6 characters, or the set of all passwords that include a dictionary word). 
\end{itemize}

We also consider the special case of \emph{singleton rules}, where our rules are $\{w_1\},\ldots,\{w_N\}$. Equivalently, we are allowed to ban or allow any individual password. 

We use $\Pr[w~\vline ~\mathcal{A}]$ to denote the probability of a password $w$ given composition policy $\mathcal{A}$. For $w \notin \mathcal{A}$ we have $\Pr[w~\vline ~\mathcal{A}]=0$. Given a set $W \subseteq \mathcal{A}$ we will also use $\Pr[W~\vline~ \mathcal{A}] = \sum_{w \in W} \Pr[w~\vline~ \mathcal{A}]$. We use $p\left(k,\mathcal{A} \right) = \max_{W \subseteq \mathcal{A}: \left|W\right|=k} \Pr[W ~\vline ~\mathcal{A}]$ to denote the probability of the $k$ most popular passwords. Intuitively, $p\left(k,\mathcal{A} \right)$ represents the probability that an adversary can successfully guess a password using $k$ attempts. To avoid cumbersome notation we sometimes use $p_1 = p\left(1,\mathcal{A}\right)$ to denote the probability of the most popular password.  Similarly, we use $p_2$ (resp., $p_k$) to denote the probability of the second (resp., $k$'th) most popular password.

%

We consider two user models that determine how users choose passwords under a given password composition policy. 

\begin{itemize}
\item \emph{The ranking model:} A ranking is simply a permutation of $\PasswordSpace$, which represents a user's password preferences. It can be represented using an ordered list $\ell_i = w_{1,i},...,w_{N,i}$; user $i$ prefers password $w_{j,i}$ to $w_{j+1,i}$ for all $j$. The ranking $\ell_i$ naturally tells us which password $i$ will pick under any composition policy $\mathcal{A}$. Specifically, $i$ will use password $w_{\mathcal{A},i} = w_{j,i}$ where  $j = \text{argmin}\{t : w_{t,i} \in \mathcal{A}\}$. Given a distribution $\mathcal{D}$ over rankings, we have 
\[\Pr\left[w~\vline \mathcal{A}\right] = \Pr_{\ell_i\sim \mathcal{D}}\left[w_{\mathcal{A},i} = w\right]  \ .\] 

\item\emph{The normalization model:} 
Let $\mathcal{D}$ be an initial distribution over $\PasswordSpace$, and let $\Pr\left[w \right] = \Pr_{x \sim \mathcal{D}}\left[w = x\right]$. If we select the composition policy $\mathcal{A}$ then the probabilities of all $w \in \mathcal{A}$ are simply re-normalized so that
\[\forall w \in \mathcal{P}, \mathcal{A} \subseteq \mathcal{P}, \Pr\left[w~\vline\mathcal{A} \right] = \frac{\Pr\left[w \right]}{\Pr\left[\mathcal{A}\right]}  \ .\]
\end{itemize}

Clearly it holds for both models that the probability of an allowed password monotonically increases as one bans more passwords. Formally, for all $w \in \mathcal{A}$ and $B \subseteq \mathcal{P}$ such that $w \notin B$ we have
\begin{equation} 
\label{claim:banningNecessaryClaim}
\Pr\left[ w ~ \vline \mathcal{A}\right]  \leq  \Pr\left[ w ~ \vline \mathcal{A} \backslash B \right]  \ . 
\end{equation}

Another important observation is that for our purposes the ranking model is more general than the normalization model. Indeed, we argue that a distribution $\mathcal{D}$ over passwords in the normalization model induces an equivalent distribution over rankings. To generate the most highly ranked password, draw a password $w_1$ from $\mathcal{D}$. Next, let $\mathcal{A}_1=\PasswordSpace\setminus\{w_1\}$, and draw the next most preferred password $w_2$, where $w_2=w$ with probability $\Pr[w\ |\mathcal{A}_1]$. In the following round we ban $w_2$ to obtain a policy $\mathcal{A}_2$, and so on, until all passwords have been banned. 

Given $k\in\mathbb{N}$, our goal is to find $S \subseteq [m]$ such that $p\left(k,\mathcal{A}_S\right) \leq p\left(k,\mathcal{A}_{S'}\right)$ for all $S' \subseteq [m]$. When $k=1$ this goal is equivalent to maximizing the minimum entropy. If $p\left(k,\mathcal{A}_S\right) \leq c\cdot p\left(k,\mathcal{A}_{S'}\right) + \epsilon$ for all $S' \subseteq [m]$ then we say that $S$ is a $(c,\epsilon)$-approximation. To simplify notation we sometimes use $c$-approximation instead of $(c,0)$-approximation.

\section{Ranking Model: Complexity Results}
In this section we consider the complexity of finding the optimal password composition policy in the more general ranking model when the organization is given complete information about users' preferences. Specifically, the organization is given the rankings $\ell_1,...,\ell_n$ of every user. 

Our first result is for the positive rules setting. Given positive rules $R_1,...,R_m$ we show that $p\left(k,\mathcal{A}_S\right)$ can be computed efficiently for constant values of $k$ (see Theorem \ref{thm:ConstantKAlg}). In fact, for the special case $k=1$ we present a very simple algorithm that suffices. Both algorithms can be easily extended to the less general normalization model. Our algorithms are based on three simple ideas: (1) Reduced Preference Lists --- each preference list $\ell_i$ can be efficiently reduced to a short (length $\leq m$) preference list $\hat{\ell}_i$. (2) Guess and Check --- start by guessing the `structure' of the optimal solution and find the resulting solution. (3) Iterative Elimination --- find the most popular password $w$ and eliminate all positive rules that contain $w$. Our sampling algorithms are based on the same core ideas.

Unfortunately, the picture is different in the negative rules even when $k$ is a constant. Given negative rules $R_1,...,R_m$ we show that it is hard to even $n^{1/3}$-approximate $p\left(1,\mathcal{A}_S\right)$. Also, for non-constant values of $k$ we show that it is hard to compute $p\left(k,\mathcal{A}_S\right)$ in the singleton rules setting, which immediately implies hardness in both the positive rules setting and in the negative rules setting. Given a stronger complexity assumption known as the Unique Games Conjecture \cite{khot2002power} it is also hard to $c_0$-approximate $p\left(k,\mathcal{A}_S\right)$ in the singleton rules setting for some constant $c_0$. However, our hardness results do not rule out the possibility of a $c$-approximation for a larger constant $c$.

\subsection{Positive Rules: Efficient Algorithm for Constant $k$} \label{subsec:algForConstantk}
We first show that $p\left(k,\mathcal{A}_S\right)$ can be computed efficiently for constant values of $k$ in the positive rules setting. In this section the organization is given positive rules $R_1,...,R_m$ as well as preference lists $\ell_1,...,\ell_n$. We assume that the organization can efficiently query the preference lists (e.g., given $S \subseteq [m]$ the organization can efficiently find $\ell_i\left(\mathcal{A}_S\right)$ --- user $i$'s preferred password given policy $\mathcal{A}_S$). 

We elaborate on the key algorithmic ideas listed above. First, we can efficiently reduce each preference list $\ell_i$ to a list $\hat{\ell}_i$ of at most $m$ passwords (Claim \ref{claim:ReducedPreferenceListClaim}). While the reduced list $\hat{\ell}_i$ is much shorter than $\ell_i$ it is still sufficient to determine user $i$'s preferred password given policy $\mathcal{A}_S$ for any $S \subseteq [m]$. We use $\hat{\PasswordSpace}$ to denote the reduced space of potential passwords. 

\begin{algorithm}
\caption{Reduce}
\begin{algorithmic}
\State {\bf Input:}
\State {Preference List: } $\ell$
\State {Positive Rules: } $R_1,...,R_m$
\State {\bf Initialize:}  $i\gets 0$, $S_0 \gets [m]$, $\hat \ell \gets$ empty ranking.
\While{$S_i \neq \emptyset$}
\State Let $w$ be $\ell\left(\mathcal{A}_{S_i}\right)$.
\State $\hat{\ell} \gets \langle \hat\ell, w\rangle$\Comment{`Append' the current most preferred password to $\hat{\ell}$}
\State $S_{i+1} \gets S_i \setminus  \left\{j ~\vline~ w\in R_j\right\}$\Comment{Deactivate all rules that contain $w$}
\State $i \gets i+1$
\EndWhile
\Return $\hat{\ell}$
\end{algorithmic}
\label{alg:ReducePreferenceLists}
\end{algorithm}

\begin{claim} \label{claim:ReducedPreferenceListClaim} Algorithm \ref{alg:ReducePreferenceLists} makes at most $m$ queries to $\ell$ and $m^2$ membership queries and outputs a reduced preference list $\hat{\ell}$ over at most $m$ passwords such that for every $S \subseteq [m]$ it holds that $\hat{\ell}\left(\mathcal{A}_S\right) = \ell\left(\mathcal{A}_S\right)$. 
\end{claim}
\begin{proof}
Clearly, the algorithm's main loop iterates at most $m$ times because for each $i$ we eliminate at least one rule (e.g., $\left|S_{i+1}\right| < \left|S_i\right|$), so the bound on queries and the length of $\hat\ell$ are immediate. (Because we assume that we can query $\ell$ efficiently Algorithm \ref{alg:ReducePreferenceLists} is also efficient.) By construction we have $\hat{\ell}(S_i) = \ell(S_i)$ for each $S_i$. Fix any $S \subseteq [m]$. Let $S_i$ be such that $S \subseteq S_i$ yet $S \not\subseteq S_{i+1}$ and let $w_i$ be the most preferred word in $\ell$ out of all words in $\bigcup_{j\in S_i}R_j$. If it is the case that $w_i \in \bigcup_{j\in S} R_j$, then $w_i$ is the most preferred word in $S$ too and we're done. Otherwise, $w_i \in \bigcup_{j\in S_i \setminus S} R_j$ which means that removing the set $\{j \in S_i ~:~ w_i \in R_j\}$ creates a set $S_{i+1}$ s.t. $S\subseteq S_{i+1}$, contradiction.
\end{proof}

Second, the ``guess and check'' idea means that our algorithm starts by guessing what the optimal solution looks like (e.g., what the $k$ most popular passwords will be in the optimal solution and what the probability of the $k$'th most popular password is). There are at most $\left(mn\right)^{O(k)}$ potential solutions to brute-force try. As we show, for each solution, it is easy to figure out which sets must be eliminated.

\begin{algorithm}
\caption{GuessAndCheck}
\begin{algorithmic}
\State {\bf Input:} 
\State {Preference Lists: } $\ell_1,...,\ell_n$
\State {Positive Rules: } $R_1,...,R_m \subseteq \PasswordSpace$
\State Integer $k$
\State {\bf Initialize: } $Candidates \gets \emptyset$
\Comment{Candidate Solutions}
\For{$i=1 \to n$}
\State $\hat{\ell}_i \gets Reduce\left(\ell_i,R_1,...,R_m\right)$ 
\EndFor
\State $\hat{\PasswordSpace} \gets \bigcup_{i=1}^n \hat{\ell}_i$.
\Comment{Reduced Password Space}
\ForAll{$(G,p)$ with $G \subseteq \hat{\PasswordSpace}$ s.t. $|G| = k$ and $p \in \{1/n,2/n,...,1\}$}

\State $S_{G,p} \gets [m]$
\While{$S_{G,p} \neq \emptyset$ and $\exists w \in \left(\hat{\PasswordSpace}\setminus G\right) \cap \mathcal{A}_{S_{G,p}}$ s.t  $\Pr\left[w~\vline~\mathcal{A}_{S_{G,p}}\right] > p$}
 \State $S_{G,p} \gets  S_{G,p} \setminus \{j~|~w \in R_j \} $\Comment{Ban $w$ because it is inconsistent with guess}

\EndWhile

\If{$\Pr\left[w~\vline~\mathcal{A}_{S_{G,p}}\right] \leq p$ for all $w \in \left(\mathcal{A}_{S_{G,p}}\setminus G\right)$} 
\State $Candidates \gets  Candidates \cup \{{S_{G,p}}\}$  \EndIf

\EndFor
\Return $\arg\min_{(G,p) \in Candidates} p\left(k,\mathcal{A}_{S_{G,p}}\right)$
\end{algorithmic}
\label{alg:ConstantKAlgorithm}
\end{algorithm}

\begin{theorem} \label{thm:ConstantKAlg}
Algorithm~\ref{alg:ConstantKAlgorithm} runs in time polynomial in $n^k$, $m^k$ and outputs a set of positive rules $S \subseteq [m]$ of positive rules such that \[p\left(k,\mathcal{A}_S \right) \leq p\left(k,\mathcal{A}_{S'} \right)\] for every other set $S' \subseteq [m]$.
\end{theorem}
\begin{proof} 
It is evident that the running time of the algorithm is $\poly(n^k,m^k)$ since we only have $O((nm)^k)$ potential solutions to try.

Let $\mathcal{A}_{S^*}$ denote an optimal solution and let $G^*$ denote the $k$ most popular passwords in this solution. Suppose we start with the correct guess ($G=G^*$ and $p$ is the probability of the $k$'th most popular password), then we claim that our algorithm must produce the optimal solution. In particular, we maintain the invariant that $\mathcal{A}_{S^*} \subseteq \mathcal{A}_{S_{G,p}}$ until we converge to the optimal solution. \ap{The last sentence seems to indicate that you want to to prove the invariant, but for the rest of the proof you assume it.} \jb{Added more explanatory text --- is the proof clear now?} Clearly, this is true initially --- before we have eliminated any passwords.  

Suppose that the invariant holds and that our algorithm bans a password $w \in \PasswordSpace \setminus G$ \ap{Should this be $\PasswordSpace-G$?}\jb{yes} by deactivating all rules in $S_{G,p}$ that contain $w$. Then by the definition of our algorithm we must have  $\Pr\left[w~\vline~\mathcal{A}_{S_{G,p}}\right] > p$. If $w\in\mathcal{A}_{S^*}$ then by Equation \eqref{claim:banningNecessaryClaim} we have \[\Pr\left[w~\vline~\mathcal{A}_{S^*}\right] \geq \Pr\left[w~\vline~\mathcal{A}_{S_{G,p}}\right] > p \ ,\]
which contradicts the choice of $G$. Therefore $w\notin\mathcal{A}_{S^*}$, so all rules that contain it are deactivated in $\mathcal{A}_{S^*}$ and the invariant still holds. By definition Algorithm \ref{alg:ConstantKAlgorithm} terminates when every password $w \in A_{S_{G,p}} \setminus G$ \ap{Should this be $w \in A_{S_{G,p}} -G$?} \jb{yes} has probability at most $p$. Because our invariant still holds we can apply Equation \eqref{claim:banningNecessaryClaim} again to get
\[ \Pr\left[G~\vline~\mathcal{A}_{S_{G,p}} \right]  \leq \Pr\left[G~\vline~\mathcal{A}_{S^*} \right] = p\left(k,\mathcal{A}_{S^*}\right)  \ . \]
Hence, $A_{S_{G,p}}$ is an optimal solution.
\end{proof}

For the special case $k=1$ the simple algorithm IterativeElimination (Algorithm~\ref{alg:IterativeElimination}) suffices. The basic idea is very simple: iteratively eliminate the most popular password $w$ by deactivating all positive rules that contain $w$. We repeat this process until no passwords remain. We claim that one of the solutions along the way was the optimal solution.

\begin{algorithm}
\caption{IterativeElimination}
\begin{algorithmic}
\State {\bf Input:}
\State {Preference Lists: } $\ell_1,...,\ell_n$
\State {Positive Rules: } $R_1,...,R_m \subseteq \PasswordSpace$
\State {\bf Initialize: } $S_0 \gets [m]$, $i \gets 0$
\While{$S_i \neq \emptyset$} 

\State $w\left(S_i\right) \gets \arg\max \left\{\Pr\left[w ~|~\mathcal{A}_{S_i}\right]  ~\vline~  w \in \mathcal{A}_{S_i}\right\}$
\Comment{$w\left(S_i\right)$ is most popular allowed pwd}
\State $S_{i+1} \gets S_i \setminus \left\{j~\vline~w\left(S_i\right) \in R_j \right\}$ 
\Comment{Deactivate all rules that contain $w\left(S_i\right)$}
\State $i \gets i+1$
\EndWhile
\Return $S_{i^*}$ where $i^* \gets \arg\min_i p\left(1,\mathcal{A}_{S_i}\right)$
\end{algorithmic}
\label{alg:IterativeElimination}
\end{algorithm}

\begin{theorem} \label{thm:AlgorithmAddingSubsets}
Algorithm~\ref{alg:IterativeElimination} outputs a set of positive rules $S \subseteq [m]$ such that 
\[ \forall S'\subseteq [m] ~, ~ ~ p\left(1,\mathcal{A}_S \right) \leq p\left(1,\mathcal{A}_{S'} \right) \  . \]
\end{theorem}
\begin{proof}
Let $T$ denote the optimal policy. Clearly if $T = [m]$ then our algorithm returns $S^* = T$ because that is the first set we try. Otherwise, $T\subsetneq [m]$. Let $S$ be the last set our algorithm considers that has the property that $T\subseteq S$. Again, if $T=S$, our algorithm returns $S$. Let $w(T)$ be the most popular word in $\mathcal{A}_T$, and because of optimality $\Pr[w(T) ~|~ \mathcal{A}_T] \leq \Pr[w(S) ~|~ \mathcal{A}_S]$.

Now, because we modify $S$ to not contain $T$ in the next iteration, then the most popular word in $S$, $w(S)$ has to belong to some rule $R_j$ where $j \in T$. Therefore $w(S) \in \bigcup_{j\in T} R_j$, and by the definition, the most popular word in $\mathcal{A}_T$ satisfies $\Pr[w(T) ~|~ \mathcal{A}_T] \geq \Pr[w(S) ~|~ \mathcal{A}_T]$. 

But observe, because $w(S) \in \bigcup_{j\in T} R_j$, we must have that $w(S)$ is at least as popular in $T$. 
Indeed, if $\ell$ is a preference list where we disallowed $\PasswordSpace \setminus \bigcup_{j\in S} R_j$ and the most preferred word is $w(S)$, then as long as we disallow more words but keep allowing $w(S)$ the word $w(S)$ remains at the top of the list. Therefore, $\Pr[w(S)~|~ \mathcal{A}_T] \geq \Pr[w(S)~|~ \mathcal{A}_S]$. Combining together all inequalities we get $\Pr[w(T) ~|~ \mathcal{A}_T]  = \Pr[w(S) ~|~ \mathcal{A}_S]$, which means our algorithm returns $S^* = S$.
\end{proof}

\subsection{Singleton Rules: Hardness for Large $k$} \label{subsec:HardessKisParameter}
Now we turn our attention to the problem of optimizing $p\left(k,\mathcal{A}_S \right)$ for large values of $k$. Theorem \ref{thm:hardnessKisParameter} says that unless $P=NP$ no polynomial time algorithm can compute $p\left(k,\mathcal{A}_S \right)$ even with singleton rules. If we are willing to make 
the Unique Games Conjecture (UGC) \cite{khot2002power} then it is hard to even $c_0$-approximate $p\left(k,\mathcal{A}_S\right)$ for some constant $c_0$. These results immediately imply hardness in both the positive and negative rules setting because these settings are a generalization of the singleton rules setting. 

\begin{theorem}\label{thm:hardnessKisParameter}
Unless $P=NP$ there is no $\poly(k,n,N)$-algorithm that gets as input an arbitrary set of $n$ preference-lists $\ell_1,...,\ell_n$ over $\PasswordSpace$ and an integer $k$, and outputs the optimal $p(k,\mathcal{A})$ in the singleton rules setting.
\end{theorem}
\begin{proof}
We prove the theorem using a reduction from the Vertex-Cover problem. Given a graph $G$ over $g$ vertices and $e$ edges and an integer $t$, we first define
\[\PasswordSpace =  \{w_u ~:~ u\in V(G)\} \cup \{w_{u,v} ~:~ (u,v)\in E(G)\}   \]
and observe that $|\PasswordSpace| = g+e$. We also construct the following $n=2e$ preference-lists, where for every edge $(u,v)\in E(G)$ we have the two lists:
\begin{eqnarray*}
&& \ell_{u,v} =  w_u, w_{u,v}, \ldots \cr
&&\ell_{v,u} =  w_v, w_{u,v}, \ldots
\end{eqnarray*}
where the choice of passwords below position $2$ is arbitrary, but both rankings must be identical from position $2$ onwards. Finally, we set $k = g+e-t-1$.

Given a policy $\mathcal{A}\subseteq\PasswordSpace$, we denote all banned words as $\mathcal{B} = \PasswordSpace\setminus \mathcal{A}$. We denote by $L_\mathcal{B}$ as the set of words that at least one user ranks first after banning all words in $\mathcal{B}$. Observe, $L_\emptyset = \{w_u ~:~ u\in V(G)\}$. Using this notation, we show this reduction indeed proves $NP$-hardness.

First, suppose $G$ has a vertex cover $C$ of size $\leq t$. Then by banning all passwords $\mathcal{B} = \{w_v ~:~ v\in C\}$ we now have $L_{\mathcal{B}} = \PasswordSpace \setminus \mathcal{B}$, because for every $(u,v)\in E(G)$ either $w_u$ or $w_v$ are banned, so the word $w_{u,v}$ appears at the top of at least one of the two lists $\{\ell_{u,v}, \ell_{v,u}\}$. Therefore, the $n$ preference-lists induce a distribution whose support contains $g+e - |\mathcal{B}| \geq g+e-t$ words, thus $p(g+e-t-1, \mathcal{A}) < 1$. 

Conversely, suppose all vertex covers of $G$ are of size at least $t+1$. Let $\mathcal{A}$ be any set of banned words. Clearly,  if $|\mathcal{B}| \geq t+1$ then the distribution induced by the $n$ preferences-lists has support of size at most $g+e-t-1$, which means that $p(g+e-t-1, \mathcal{B})=1$. Otherwise, $|\mathcal{B}| \leq t$, and we denote the set of vertices $C = \{v : w_v \in \mathcal{B}\}$. Observe, since any vertex cover of $G$ must contain $\geq t+1$ vertices, then there has to be at least $t+1-|C|$ edges that $C$ does not cover (since we can always complete $C$ to a vertex cover by adding one vertex from each uncovered edge). Therefore, there have to be at least $t+1-|C|$ words that do not appear at the top of any preference list. We conclude that the distribution induced by the $n$ preference-lists has a support of size at most 
\[ |L_{\mathcal{B}}| = g - |C| + e - (t+1 - |C|) \leq g+e-t-1\]
thus $p(g+e-t-1, \mathcal{A})=1$.
\end{proof}

From the same reduction described in Theorem~\ref{thm:hardnessKisParameter} we get $UGC$-hardness of approximation. While there are sub-exponential time algorithms to solve the Unique Games problem \cite{arora2010subexponential}, there are no known polynomial time algorithms. Many famous approximation hardness results are based on the Unique Games Conjecture (e.g., $2-\epsilon$ hardness for vertex cover \cite{khot2008vertex}). Our reduction relies on a result in~\cite{AustrinKS11}, which says that vertex cover is hard to approximate up to a (say) $1.5$-factor even on bounded degree graphs. Because we start with a bounded degree graph we can argue that each password in our reduction appears at the top of at most $d$ preference-lists for some constant $d$. See the appendix for a formal proof.
\newcommand{\thmUGCHardness}{There exists a constant $c>1$ such that it is $UGC$-hard for a $\poly(n,N,k)$-time algorithm to $c$-approximate the optimal $p(k,\mathcal{A})$ in the singleton rules setting and the rankings model.}
\begin{theorem} \label{thm:UGCHardness}
\thmUGCHardness
\end{theorem}

\subsection{Negative Rules: Hardness of Approximation for $k=1$} \label{subsec:HardnessOfApproximatingOptimalCombinationOfSubsets}


We next turn to negative rules, where we show that the problem is extremely difficult even for $k=1$. Though the proof appears in the appendix, it is quite interesting and we encourage the reader to take a look.


\begin{theorem}\label{thm:HardnessOfApproximationOptimalCombinationOfSubsets}
Let $\epsilon > 0$. Unless $P=NP$ there is no polynomial time algorithm (in $N,n,m$) that approximates $\min_{S\subseteq [m]} p(1,\mathcal{A}_S)$ to a factor of $n^{1/3-\epsilon}$ in the negative rules setting and the rankings model. \ap{See below, possibly change bound to $\Omega(n^{1/3})$}
\end{theorem}

\section{Normalization Model: Complexity Results} \label{sec:NormalizedProbabilities}

In this section we focus on complexity results for the normalization model. Here the structure of the input to our problem is a bit different: For each password $w\in\PasswordSpace$ we are given the probability $\Pr[w]$ that $w$ is selected by a random user when $\mathcal{A}=\PasswordSpace$. Note that now we can give the distribution explicitly because it requires $N$ numbers (whereas a distribution over rankings requires $N!$ numbers). This distribution induces a distribution over $\PasswordSpace$ for any password composition policy $\mathcal{A}$ by normalizing probabilities, as explained in Section~\ref{sec:Models}.

Because the normalization model is a special case of the ranking model our algorithms for the ranking model can also be applied in the normalization model. The question is whether or not the hardness results carry over.  


We first consider the singleton rules setting with large $k$, and show that that we can compute $\arg\min_{\mathcal{A} \subseteq \mathcal{P}} p\left(k,\mathcal{A}\right)$ in polynomial time in $N$ (Theorem \ref{thm:NrmPrbEfficientAlgAnyK}). This result separates the normalization model from the ranking model (e.g., compare Theorems \ref{thm:NrmPrbEfficientAlgAnyK} and \ref{thm:hardnessKisParameter}). However, it does not extend to the positive rules setting. In fact, we show that optimizing $p\left(k,\mathcal{A}_S\right)$ is NP-Hard when $k$ is a parameter (Theorem \ref{thm:HardnessAddingSubsetsLargeK}).

With negative rules $R_1,...,R_m$ we show that it is hard to $c_0$-approximate  $\arg\max_{S \subseteq [m]} p\left(1,\mathcal{A}_S\right)$ (Theorem \ref{thm:NrmPrbHardnessOptimizeP1}). However, we cannot rule out the possibility of an efficient $c$-approximation algorithm for some constant $c$ in the normalization model (recall that Theorem \ref{thm:HardnessOfApproximationOptimalCombinationOfSubsets} ruled out the possibility of a $c$-approximation algorithm in the ranking model for any $c$).

\subsection{Singleton Rules: Efficient Algorithm for large $k$}
\label{subsec:PolyAlgForAnyKNormalizedProbabilities}

We present SortAndOptimize --- an efficient algorithm to optimize $p\left(k, \mathcal{A}\right)$ in the singleton rules setting for {\em any} value of $k$. 
The key intuition behind our algorithm is that if $w_1 \in \PasswordSpace$ is the most likely password then $w_1$ will remain the most likely allowed password unless we ban it --- a property that does not hold in the rankings model. A formal proof of Theorem \ref{thm:NrmPrbEfficientAlgAnyK} can be found in the appendix.
\newcommand{\thmNrmPrbEfficientAlgAnyK}{For every $k$, Algorithm~\ref{alg:SortAndOptimize} computes $\arg \min_{\mathcal{A}} p\left(k,\mathcal{A} \right)$ in the singleton rules setting of the normalized probabilities model, in time $O(N\log(N))$.}
\begin{theorem} \label{thm:NrmPrbEfficientAlgAnyK}
\thmNrmPrbEfficientAlgAnyK
\end{theorem}

\begin{algorithm}
\caption{SortAndOptimize}
\begin{algorithmic}
\State {\bf Input:}
\State {Password space } $\PasswordSpace$ and a probability distribution over $\PasswordSpace$.
\State {Integer} $k$.
\State {\bf Sort} the words in $\PasswordSpace$ from highest to lowest probability, $w_1, w_2, \ldots, w_N$.\\
\Return the set $\mathcal{A}_i = \{ w_j ~:~ j \geq i\}$, where $i$ minimizes the ratio \[ p(k,\mathcal{A}_i) =  \frac {\sum_{i \leq j \leq i+k} \Pr[w_j]} {\sum_{j \geq i} \Pr[w_j]} \]
\end{algorithmic}
\label{alg:SortAndOptimize}
\end{algorithm}

\subsection{Negative Rules: Hardness for $k=1$} \label{subsec:HardnessOptimizeP1BanningNrmPrb}


We next prove an inapproximability result that is somewhat weaker than the one that we obtained for the more general ranking model.

\begin{theorem}\label{thm:NrmPrbHardnessOptimizeP1}
There exists some constant $c_0>1$ such that unless $NP=BPP$ no polynomial time algorithm (in $n, N, m$) can $c_0$-approximate $\min_{S \subseteq [m]} p\left(1,\mathcal{A}_S\right)$ in the negative rules setting and the normalization model.
\end{theorem}

We will require the following construction; the proof is given in the appendix. 

\begin{lemma}\label{clm:FamilyOfSetsWithLargeUnions}
Fix $m$ and $s$ such that $m\geq s$. There exists a domain $D$ of size $\Theta(s^2\log(m))$ and a family of $m$ sets, $F_1, F_2, \ldots, F_m\subseteq D$, such that each set in the family contains $\tfrac {|D|}{2s}$ elements, and for every $C\subseteq [m]$ of size $|C| \leq s$, we have that the size of the union $\left| \bigcup_{i\in C} F_i \right| \geq \tfrac {|D|}{2s} \tfrac{|C|} 4$. This domain can be constructed in randomized $\poly(s,m)$ time.
\end{lemma}

That is, each set in this family contains exactly the same fraction of the domain, and furthermore --- any union of $|C| \leq s$ sets has the property that its cardinality is proportional to $\Omega(|C|) |F_i|$. 

\begin{proof}[of Theorem~\ref{thm:NrmPrbHardnessOptimizeP1}]
We reduce from Set-Cover --- one of the classic $NP$-Complete problems \cite{karp1972reducibility}. We are given sets $S_1,...,S_m \subseteq U$, universe $U = \{1,...,g\}$, and an integer $t \leq m$, and we are asked whether there is a set $C \subseteq [m]$ of size $\leq t$ such that $U = \bigcup_{i \in C} S_i$.

It is a known fact that there exist Set-Cover instances, with $(g,m,t)$ all polynomially dependent of each other, that are hard to approximate to a factor of $c\ln n$~\cite{AlonMS06}. That is, on this particular family of instances, it is $NP$-hard to distinguish whether there exists a cover of size $t$ or all covers have size $(1-\epsilon)c\cdot t\ln n$. 

We now describe the reduction. Given a $(g,m,t)$-Set Cover instance, we set $s = c\cdot t\ln g = \Theta(t\ln t)$ and construct a domain $D$ and $m$ sets $F_1, F_2, \ldots, F_m\subseteq D$ as in Lemma~\ref{clm:FamilyOfSetsWithLargeUnions}. We then create the following password-banning instance. First $\mathcal{P}$ is the union of $D$ with additional disjoint $g$ words denoted $w_1,...,w_g$. Now, for each set $S_i$ in the Set-Cover we add a rule $R_i$ where $R_i = \{w_j\}_{j\in S_i} \cup F_i$. Finally, we set the words' probabilities as follows. Fixing some arbitrarily small $\delta > 0$, we set for every $i$ the probability $\Pr[w_i] = \tfrac{1-\delta}{g}$, and for every $x\in D$ we set the probability $\Pr[x] = \tfrac \delta{|D|}$.

Without loss of generality we can assume that $|D| \geq 100 g$ (because, for example, we can take $100g$ copies of the original $D$). Therefore, any policy that bans all of $\{w_1, w_2, \ldots w_g\}$ yet leaves a constant (say $>1/10$) fraction of $D$ has $p_1 \leq 10/|D|$, whereas any policy that keeps even one of the words in $\{w_1, w_2, \ldots, w_g\}$ has $p_1 \geq 1/(2g)$. Therefore, if the Set-Cover instance has a cover of size $\leq s = \Theta(t\ln g)$, then a $c_0$-approximation of the optimal banning-policy must find a cover for $\{w_1, w_2, \ldots, w_g\}$. We will assume from now on that our Set-Cover instance is such that it has a cover of size $\leq s$. \ap{Please clarify why we can make this assumption} (Indeed, if $s > t\log(t)$ then the instance is no longer $NP$-hard, since the greedy algorithm must return a cover of size $>t \log(t)$ which causes us to deduce that the optimal cover must have size $>t$.) \os{Clearer?}

So now, suppose our Set-Cover instance has a cover of size $t$. Then the respective union of rules bans every password in $\{w_1, w_2, \ldots, w_g\}$ and no more than $\tfrac t {2s} |D|$ words of $D$ (we get an upper bound by multiplying the size of each set by the number of sets).\ap{Added the parenthetical comment, please check}\os{True. Do you want to assume the reader cannot figure this one on her own?}\jb{My two cents. As a reader I usually appreciate clarifying remarks --- even if I didn't need them} This leaves a collection of $\left(1-\tfrac t {2s}\right) |D|$ equally likely words, so $p_1 = \left(1-\tfrac t {2s}\right)^{-1}|D|^{-1} = (1-O(1/\log(g)))^{-1}|D|^{-1} = (1+o(1))|D|^{-1}$.  In contrast, if all covers of our Set-Cover instance have size $s' \geq c\cdot t\ln(g)$ (where, because we assume some cover has size $\leq s$, we have $s'\leq s$,) then any collection of rules that bans all words in $\{w_1, w_2, \ldots, w_g\}$ must also ban at least $\tfrac {s'}{8s} |D|$ words out of $D$. This leaves at most $(1-\Omega(1))|D|$ words in $D$ and so $p_1 \geq (1-\Omega(1))^{-1}|D|^{-1}$. Denoting the latter constant as $c_0^{-1}$, we have that any $c_0-\epsilon$ approximation of the optimal banning-policy indicates the existence of a cover of cardinality $< c\cdot t\ln(g)$. 
\end{proof}

\subsection{Positive Rules: Hardness of Approximation for  Large $k$}

While we can show that it is possible to optimize $p\left(k,\mathcal{A}\right)$ in the singleton rules setting our result does not extend to the more general positive rules setting. We are able to show that it is NP-Hard to compute $\arg\min_{S \subseteq [m]} p\left(k,\mathcal{A}_S \right)$. However, our reduction does not imply approximation hardness so we cannot rule out the existence of a PTAS.

\newcommand{\thmHardnessAddingSubsetsLargeK}{Unless $P=NP$ there is no polynomial time algorithm (in $N,m,n$) which outputs  $\arg\min_{S \subseteq [m]} p\left(k,\mathcal{A}_S \right)$ in the positive rules setting and the normalization model.}

\begin{theorem} \label{thm:HardnessAddingSubsetsLargeK}
\thmHardnessAddingSubsetsLargeK
\end{theorem}

The theorem's proof is relegated to the appendix. 
 %
%
%
%
%

\section{Efficient Sampling Algorithms} \label{sec:efficientsampling}

In a sense, our complexity results are not ``realistic'', and in particular in the ranking model our positive algorithmic results assume access to each user's full preferences. Moreover, some algorithms are allowed to run in polynomial time in the number of passwords $N$, which can be huge. In this section we use our complexity results as guidelines in the design of practical sampling algorithms. 

In more detail, we are given oracle access to rules $R_1,...,R_m$ (e.g., we can ask whether or not a password $w \in R_i$) and we are allowed to sample from the distribution induced by the password composition policy $\mathcal{A}_S$ for any $S \subseteq [m]$. Less formally, a sample is equivalent to asking a random user what her favorite password is given the current policy. 

We will work in the more general ranking model, so there is essentially only one positive result we can build on: Theorem~\ref{thm:ConstantKAlg}, a polynomial time algorithm for constant $k$ in the positive rules setting. When adapting this algorithm to the sampling setting, we cannot expect it to work perfectly due to the inherent uncertainty of this domain. Instead we expect the algorithm to find an $\epsilon$-optimal password composition policy with probability at least $1-\delta$, for any given $\epsilon$ and $\delta$. Crucially, the number of samples must not depend on the number of passwords $N$, and must have a polynomial dependence on the other parameters. 

Formally, we let $S^* \subseteq [m]$ denote the optimal collection of positive rules to activate (for all $S \subseteq [m]$, $p\left(1,\mathcal{A}_{S^*}\right) \leq p\left(1,\mathcal{A}_{S}\right)$). Our goal is to find a $(1,\epsilon)$-approximation $S \subseteq [m]$ to $p\left(1,\mathcal{A}_{S^*}\right)$, that is, $S$ such that $p\left(1,\mathcal{A}_{S}\right) \leq p\left(1,\mathcal{A}_{S^*}\right)+\epsilon$, with probability $1-\delta$.

We first present Algorithm \ref{alg:SampleAndEliminate} that achieves our goal for $k=1$; this algorithm is an adaptation of Algorithm~\ref{alg:IterativeElimination}. 

\begin{algorithm}
\caption{SampleAndEliminate}
\begin{algorithmic}
\State {\bf Positive Rules: } $R_1,...,R_m$ 
\State {\bf Input: } $\epsilon$, $\delta$
\State {\bf Initialize: } $S_0 \gets [m], i \gets 0$
\State $s \gets \frac{100}{\epsilon^2} \log \left(\frac{4m}{\epsilon\delta} \right)$
\While{$S_i \neq \emptyset$}
\State {\bf Sample: } Draw samples $w_1,...,w_s$ according to the distribution $\Pr\left[w~\vline~\mathcal{A}_{S_i}\right]$
\State $W \gets \left\{w_1,...,w_s\right\}$
\State $s_w \gets \left| \left\{j~\vline w_j = w \right\} \right|$ for each $w \in W$.
\State $w^* \gets \arg\max \left\{ s_w ~\vline~w \in W \right\}$
\Comment{$w^*$ is the most frequently sampled password}
\State $\hat{p}_i \gets \frac{s_{w^*} }{ s }$ 
\Comment{$\hat{p}_i$ is our estimation of $\Pr\left[w^*~\vline~\mathcal{A}_{S_i} \right]$}
\If{$\hat{p}_i \leq \epsilon/2$} 
\Return $S_i$
\Comment{The current solution is already sufficiently good}
 \Else 
 \State $S_{i+1} \gets S_i - \{j ~\vline ~w^* \in S_j\}$
 \Comment{Deactivate all rules that contain $w^*$}
  \State $i \gets i+1$

 \EndIf
\EndWhile
\Return $S_{i^*}$ where $i^* = \arg\max \left\{ \hat{p}_j \vline j \leq m \right\}.$ 
\end{algorithmic}
\label{alg:SampleAndEliminate}
\end{algorithm}

\begin{theorem}
Algorithm \ref{alg:SampleAndEliminate} runs in polynomial time in $m, 1/\epsilon,1/\delta$, requires $O\left(m \log \left(m/\delta \right)/\epsilon^2 \right)$ samples and returns a $(1,\epsilon)$-approximation $S \subseteq \{1,...,m\}$ of $p\left(1,\mathcal{A}_{S^*}\right)$ with probability at least $1-\delta$. 
\end{theorem}

\begin{proof}
Let
\[BAD_i = \left\{\exists w \in \mathcal{A}_{S_i}~\vline~ \left|\frac{s_w}{s} - \Pr\left[w ~\vline~\mathcal{A}_{S_i} \right]\right| \geq \epsilon/2 \right\} \ , \] \ap{should the $\epsilon$ be there on the left hand side?} \jb{nope that was a typo} denote the event that our probability estimates are off during iteration $i$. Claim \ref{claim:SamplingBadEvent} bounds the probability of any bad event. The proof of Claim \ref{claim:SamplingBadEvent} can be found in the appendix. The proof involves bucketing the passwords based on their probability, applying Chernoff Bounds to upper bound the probability of a bad estimate for our passwords in each bucket, and repeatedly applying union bounds. 

\begin{claim}\label{claim:SamplingBadEvent}
$ \Pr\left[\exists i, BAD_i\right] \leq \delta \  . $ 
\end{claim} 

For the rest of the analysis we assume that no bad event occurs. Let $p^* = \min_{S \subseteq [m]} p\left(1,\mathcal{A}_S\right)$ and suppose that $A_{S^*} \subseteq A_{S_i}$. Clearly, this is true when $i=0$.  If $\hat{p}_i \geq \epsilon/2 + p^*$ then $\Pr\left[w^*~\vline~\mathcal{A}_{S^*}\right] \geq \Pr\left[w^*~\vline~\mathcal{A}_{S_i}\right] > p^*$ so that $w^* \notin A_{S^*}$. Hence, $A_{S^*} \subseteq A_{S_{i+1}}$ and the property is maintained for at least one more iteration. If instead $\hat{p}_i < \epsilon/2 + p^*$ then we have $\hat{p}_{i^*} \leq \hat{p}_i \leq p^* + \epsilon/2$ so for each $w \in \mathcal{A}_{S_{i^*}}$ we have $\Pr\left[w ~\vline~\mathcal{A}_{S_{i^*}} \right] \leq p^* + \epsilon$. We conclude that the solution $S_{i^*}$ is a $(1,\epsilon)$-approximation.
\end{proof}

We next explain how to extend Algorithm \ref{alg:ConstantKAlgorithm} to $(1,\epsilon)$-approximate the optimal $p\left(k,\mathcal{A}_S\right)$ for any constant $k$.

\begin{theorem}
There is an algorithm which runs in polynomial time (in $m, 1/\epsilon$, $\delta$), takes a polynomial number of samples, and returns a $(1,\epsilon)$-approximation $S \subseteq [m]$ of $p\left(k,\mathcal{A}_{S^*}\right)$ with probability at least $1-\delta$.
\end{theorem}

\begin{proof}[sketch]
To extend Algorithm \ref{alg:ConstantKAlgorithm} to $(1,\epsilon)$-approximate $p\left(k,\mathcal{A}_S\right)$ for constant $k$ we need one more idea. We cannot simply obtain a reduced password space $\hat{P}$ by reducing preference lists because we can only sample from our distribution. Notice that for any $S \subseteq [m]$ such that $i \in S$ we have $\Pr\left[ w ~\vline~ \mathcal{A}_S \right] \leq  \Pr\left[w ~\vline~ \mathcal{A}_{\{i\}}\right]$ so to obtain a $(1,\epsilon)$-approximation it is sufficient to limit our attention to passwords in the following set
\[\hat{P} =  \left\{w~\vline~\exists i, \Pr\left[w~\vline~\mathcal{A}_{\{i\}} \geq \frac{\epsilon}{k} \right] \right\}  \ .\]
We can obtain a superset of $\hat{P}$ by sampling. For each positive rule $R_i$ we draw $s$ independent samples from the distribution $\mathcal{A}_{\{i\}}$ and set
\[T_i = \left\{ w~\vline~\frac{s_w}{s} > \frac{\epsilon}{2k} \right\} \ . \]
Intuitively, a password $w $ is included in $T_i$ if and only if our estimated probability is sufficiently large. Let $T = \bigcup_i T_i$. For a sufficiently large sample size $s=O\left(poly\left(m,k,1/\epsilon,1/\delta \right) \right)$ we can apply Chernoff Bounds to argue that with probability $1-\delta$  (1) $\left|T\right|$ is small, i.e., $O\left(poly\left(m,k,1/\epsilon,1/\delta \right) \right)$, and (2) $T \supset \hat{P}$. \ap{The claim on $|T|$ is unconditional, right?} \jb{correct}
\end{proof}

\section{Experiments} \label{sec:experiments}

To demonstrate how our ideas could apply in a real-world scenario, we simulated runs of Algorithm \ref{alg:SampleAndEliminate} by sampling with replacement from the RockYou leaked password set~\cite{rockYouPasswords}. The set contains over 32 million passwords with a frequency distribution similar to that of many other password sets~\cite{bonneau2012science}. Note that all results presented here are limited by the dataset and assume the normalization model. Working in the normalization model is crucial because we cannot ask the RockYou users for their preferred password under a specific policy; an initial distribution over $\PasswordSpace$ --- which is available to us --- is sufficient though, because it induces a distribution for any policy $\mathcal{A}$. 

We selected 21 positive rules that mirror commonly used password composition rules that are used in practice, and looked at sample sizes $s$ of 100, 500, 1000, 5000, and 10000. The rules included length requirements, character class requirements, combinations of requirements, a dictionary check, etc. (See Appendix~\ref{subsec:experimentsetup} for a complete listing of the rules we selected.) For each run with a particular value of $s$, the algorithm returns a policy $\mathcal{A}_{S}$ for which we can measure $p\left(1,\mathcal{A}_{S}\right)$ in the original dataset and compare with the optimal $p\left(1,\mathcal{A}_{S^*}\right)$, determined from running Algorithm \ref{alg:IterativeElimination} on the original dataset. We performed 500 runs for each of the five values of $s$.

To gain an understanding of how policies based on negative rules perform, we took the complement of the 21 positive rules selected above to get 21 negative rules. We then determined the optimal negative rules policy by calculating $S^* = \arg \min_{S \subseteq [m]} p\left(1,\mathcal{A}_{S}\right)$ via brute-force. This was required because we have no equivalent to Algorithm \ref{alg:IterativeElimination} for negative rules. With this baseline in hand, we designed two na\"{\i}ve algorithms, similar in spirit to Algorithm \ref{alg:SampleAndEliminate}. There are multiple ways to discard a password in the negative rules setting, and one algorithm makes this decision randomly while the other bans the smallest subset as determined from the current sample. Again, 500 runs were performed for each $s \in \{100, 500, 1000, 10000, 50000\}$.

\subsection{Baselines}

\begin{table}[b]
\centering
 \tbl{Baseline probabilities for the RockYou dataset}{
\begin{tabular}{| l | @{\hskip 8pt}l@{\hskip 8pt} | p{2in} |}
\hline
{\textbf{Baseline}} & {{\boldmath$p\left(1,\mathcal{A}_{S}\right)$}} & {{\boldmath$S$}} \\
\hline
Mean across negative rules policies & 1.3\e{-2} & \\
\hline
Mean across positive rules policies & 1.0\e{-2} & \\
\hline
All passwords allowed (no policy) & 9.2\e{-3} &  \\
\hline
One positive rule ($S \in \{1,...,m\}$) & 6.8\e{-4} & 8 chars, 1 upper, 1 digit \\
\hline
Optimal policy with positive rules & 4.4\e{-4} & 14 chars OR 2 symbols OR 8 chars, 1 upper, 1 digit \\
\hline
Optimal policy with negative rules & 1.4\e{-4} & 10 chars AND 2 digits AND 1 symbol AND 1 lowercase AND not in dictionary \\
\hline
\end{tabular}}
\label{tab:Baselines}
\end{table}

\begin{table}[tb]
\centering
 \tbl{Performance of Sampling Algorithms with Positive Rules}{
\begin{tabular}{| l | l | l | l |}
\hline
{\textbf{Sample Size}} & {mean {\boldmath$p\left(1,\mathcal{A}_{S}\right)$}} & {min {\boldmath$p\left(1,\mathcal{A}_{S}\right)$}} & {\textbf{\% Optimal}} \\
\hline
100 & 6.8\e{-3} & 1.2\e{-3} & \\
\hline
500 & 9.7\e{-4} & {\boldmath$4.4\e{-4}$} & 2\%\\
\hline
1000 & 9.5\e{-4} & {\boldmath$4.4\e{-4}$} & 10\%\\
\hline
5000 & 6.0\e{-4} & {\boldmath$4.4\e{-4}$} & 14\%\\
\hline
10000 & 5.7\e{-4} & {\boldmath$4.4\e{-4}$} & 19\%\\
\hline
\end{tabular}}
 \label{tab:PositiveSamplingPerformance}
\end{table}

\begin{table}[tb]
\centering
 \tbl{Performance of Sampling Algortihms with Negative Rules}{
\begin{tabular}{| l | l | l || l | l |}
\cline{2-5}
\multicolumn{1}{c|}{} & \multicolumn{2}{|c||}{\textbf{Random Decision}} & \multicolumn{2}{|c|}{\textbf{Ban Smallest}} \\
\hline
{\textbf{Sample Size}} & {mean {\boldmath$p\left(1,\mathcal{A}_{S}\right)$}} & {min {\boldmath$p\left(1,\mathcal{A}_{S}\right)$}} & {mean {\boldmath$p\left(1,\mathcal{A}_{S}\right)$}} & {min {\boldmath$p\left(1,\mathcal{A}_{S}\right)$}} \\
\hline
100 & 6.8\e{-3} & 1.2\e{-3} & 7.2\e{-3} & 2.3\e{-3} \\
\hline
500 & 4.4\e{-3} & 6.3\e{-4} & 9.0\e{-3} & 2.3\e{-3}  \\
\hline
1000 & 4.3\e{-3} & 4.5\e{-4} & 8.6\e{-3} & 2.3\e{-3}  \\
\hline
5000 & 6.3\e{-3} & 4.5\e{-4} & 9.2\e{-3} & 9.2\e{-3} \\
\hline
10000 & 7.2\e{-3} & 4.5\e{-4} & 9.2\e{-3} & 9.2\e{-3} \\
\hline
\end{tabular}
}
\label{tab:NegativeSamplingPerformance}
\end{table}

We examined several baselines for comparison with our algorithm. Table~\ref{tab:Baselines} shows these baselines, the probability of the most frequent password in the resulting policy, and the optimal policy as a union or intersection of rules (for clarity, the complement of the union of negative rules is shown as the intersection of positive rules).

As shown in Table~\ref{tab:Baselines} from the means across policies, randomly selecting a policy from the power set of rules can be worse than having no policy. \sk{Evidence for the need for a principled approach --- maybe mention in discussion?} The ``one rule maximum'' baseline was selected because, if decided based on sampling, only $m$ distributions need be sampled. Our efficient algorithm requires the same amount of sampling, but can find the optimal policy over $S \subseteq [m]$ rather than $S \in \{1,...,m\}$. Also of interest is the optimal policy with negative rules, which is over 3x better than the optimal policy with positive rules. However, as shown in the following section, the performance of our sampling algorithms with negative rules was far worse than in the positive rules setting.

\subsection{Performance}

In the positive rules setting (see Table~\ref{tab:PositiveSamplingPerformance}), the algorithm performed extremely well even at moderate sample sizes. The average policy selected with $s = 500$ was almost 10x better than having no policy. At $s = 1000$, the optimal policy was found 10\% of the time (50 out of 500 times).

In the negative rules setting (see Table~\ref{tab:NegativeSamplingPerformance}), however, neither algorithm found the optimal policy. The ``Ban Smallest'' heuristic, when faced with a choice between multiple subsets that contain the most likely password, decides to ban the smallest available subset, disrupting the space the least. This might seem like an intuitively good choice but, in fact, it fails to find a better policy than the empty set at large sample sizes. The randomized algorithm does better (it cannot actually do worse) but still has much worse average case performance than using our efficient algorithm with positive rules.

\section{Discussion}
\label{sec:disc}

We conclude by discussing some key points. 

\medskip
\noindent\textbf{Where do the rules comes from?} Throughout the paper we have assumed that the rules (whether positive or negative) are given as part of the input; it is not up to us to find these rules. Our experiments indicate that a collection of intuitive and practical rules can already give very good results on real data. However, the question of deciding which rules should be added to our collection is outside the scope of this paper. Much like the problem of feature selection, it is an interesting problem with real-life implications, which we suspect will be very difficult in practice.

\smallskip
\noindent{\bf Alternate policy goals.} Our goal \cite{boztas1999entropies} has been to minimize $p\left(k, \mathcal{A}_S \right)$. Intuitively, $p\left(k, \mathcal{A}_S \right)$ represents the probability that an adversary with no background knowledge can successfully guess the password of a randomly selected user in $k$ tries. A small value of $k$ optimizes security guarantees against an online guessing attack in which the adversary is locked out after $k$ failed attempts to login. A much larger value of $k$ (e.g., $2^{32}$) is necessary to optimize security against an adversary who has obtained the cryptographic hash of a password and is able to mount a brute-force dictionary attack \cite{seeley1989password}. However, the optimal solutions for $p\left(1, \mathcal{A}_S \right)$ and $p\left(2^{32}, \mathcal{A}_S \right)$ might be completely different. One stronger goal that we might hope to achieve is to optimize both goals simultaneously. More formally, can we find a policy $S \subseteq[m]$ such that for every $S' \subseteq[m]$ and every $k \leq N$ we have $p\left(k,\mathcal{A}_S\right) \leq c\cdot p \left(k,\mathcal{A}_{S'}\right)$ for some constant $c$? Unfortunately, the answer is no. For any constant $c$ this universal approximation goal is impossible to satisfy in the ranking model (see Theorem \ref{thm:ImpossibilityOfConstantFactorUniversalAppx}). 
 
 Other natural goals include $\alpha$-work factor \cite{pliam2000incomparability} and a refinement called $\alpha$-guesswork \cite{bonneau2012science} (e.g., maximize the total number of guesses needed to compromise $\alpha$-fraction of the accounts). While $\alpha$-guesswork is an useful metric to analyze the security of 70 million Yahoo passwords \cite{bonneau2012science}, it may not be a desirable optimization goal for the organization because it might allow the adversary to crack up to $\alpha-\epsilon$-fraction of the accounts with relatively few guesses. 
 
Another interesting direction is to account for an adversary with basic background information about the user (e.g., e-mail address, username, birthday). It may not always be realistic to assume that the adversary has no background knowledge because the adversary can often easily obtain some background knowledge about a user by searching for publicly available information on the internet. One approach might be to design a rule $R$ to specify different passwords for different users (e.g., the set of passwords that contain the username or birthday of the user). \\

\smallskip
\noindent {\bf Open Questions.} 
While we were able to prove several hardness results about finding the optimal password composition policy in the negative rules setting, it is possible that these hardness results could be circumvented by making mild (hopefully realistic) assumptions about the underlying password distribution or the rules $R_1,...,R_m$. Are there efficient algorithms to optimize $p\left(k,\mathcal{A}_S\right)$ in the negative rules setting given realistic assumptions? It is also possible that mild realistic assumptions could be used to circumvent the impossibility result of Theorem \ref{thm:ImpossibilityOfConstantFactorUniversalAppx}, and design a universal approximation algorithm. 

There are also several interesting technical questions that remain open: 
\begin{enumerate}
\item Normalization model with negative rules: Can we efficiently $c$-approximate $p\left(1,\mathcal{A}_{S^*}\right)$ for any constant $c$? Is there a sub-exponential algorithm (in $m$) to compute $p\left(1,\mathcal{A}_{S^*}\right)$? 
\item Ranking model with positive rules: Can we efficiently $c$-approximate $p\left(k,\mathcal{A}_{S^*}\right)$ for some constant $c$ when $k$
is a parameter?
\end{enumerate}

\smallskip
\noindent\textbf{The future.} There is a real need for a principled approach to optimizing password composition policies. We have taken a first step in this direction by providing an intuitive theoretical model and showing that it leads to algorithms that perform well on real data. We can only hope that our work will spark a fundamentally new interaction between theory and practice in passwords research.

\bibliography{passwords}

\newpage

\appendix
\section{Missing Proofs}

\begin{remindertheorem}{ \ref{thm:UGCHardness}}
\thmUGCHardness
\end{remindertheorem}

\medskip

\begin{proofof}{Theorem \ref{thm:UGCHardness}} 
We begin with a construction of a bounded degree graph which is hard approximate up to a (say) $1.5$-factor. As shown in~\cite{AustrinKS11}, for every constant $d$ there exists a family of $d$-regular graphs for which it is $UGC$-hard to determine whether there exists a vertex cover of size $t$, or all vertex-covers have size at least $\left(2-O(\log\log(d)/\log(d))-\epsilon\right) t$. Fixing $d$ to be a large enough constant such that this factor is $>1.5$, we now reduce this family of instances to a password problem using the exact same construction as in the proof of Theorem~\ref{thm:hardnessKisParameter}, with the exception that we set $k = g+e - (1.5-\epsilon) t$.

Observe, for this family of instances, $e = O(g)$ so $|\PasswordSpace| = O(g)$, but also the size of the optimal vertex-cover has to be $\Theta(g)$ (at most $g$ and at least $g/d)$. Furthermore, each password appears at the top of at most $d$ preference-lists. Therefore, by allowing $\mathcal{A}$ and banning $\mathcal{B}=\PasswordSpace\setminus \mathcal{A}$, we not only have a distribution whose support is of size $|L_\mathcal{B}|$, but it also holds that the probability of each word in $L_\mathcal{B}$ is $\Omega(1/|L_{\mathcal{B}}|)$.

Therefore, if the graph has a vertex-cover $C$ of size $t$, then by banning all words $\mathcal{B} = \{ w_u ~:~ u\in C\}$ we have that the $n$ preference-lists induce a distribution over $|L_\mathcal{B}| \geq g+e-t$. Since we set $k = g + e - (1.5-\epsilon)t$ we have that the set of most uncommon passwords contain at least $(0.5 -\epsilon)t = \Omega(|L_{\mathcal{B}}|)$ words, each with $\Omega(1/|L_{\mathcal{B}}|)$ probability, thus $p(k,\mathcal{A}) = 1- \Omega(1)$. (And, in particular, for the optimal policy $\mathcal{A}^*$ we have $p(k,\mathcal{A}^*) = 1 - \Omega(1)$.)

In contrast, applying the same argument from the proof of Theorem~\ref{thm:hardnessKisParameter}, we have that if $G$ has all vertex-covers of size $>(1.5-\epsilon)t$ then $p(k,\mathcal{A}) = 1$. The $O(1)$-hardness of approximation follows.
\end{proofof}

\begin{remindertheorem}{\ref{thm:HardnessOfApproximationOptimalCombinationOfSubsets}}
Let $\epsilon > 0$. Unless $P=NP$ there is no polynomial time algorithm (in $N,n,m$) that approximates $\min_{S\subseteq [m]} p(1,\mathcal{A}_S)$ to a factor of $n^{1/3-\epsilon}$ in the negative rules setting and the rankings model. \ap{See below, possibly change bound to $\Omega(n^{1/3})$}
\end{remindertheorem}

\medskip

\begin{proofof}{Theorem~\ref{thm:HardnessOfApproximationOptimalCombinationOfSubsets}}
Fix $\epsilon>0$. Our reduction is from the Max-Independent-Set problem, which is known to be hard to approximate up to a factor of $n^{1-\epsilon}$~\cite{Hastad96}. 
We are given a graph $G$ with $g$ vertices and $e$ edges, and we must determine whether the size of $G$'s largest independent set is $g^{1-\epsilon}$ or $g^{\epsilon}$.  

Given a Max-Independent-Set instance, we denote $K = g^{\epsilon}$ and create the following password policy instance, which is composed out of the following set of possible words:
\begin{eqnarray*}
\PasswordSpace &=& \{A_1,...,A_K\} \cup \{B_1,...,B_g\}  \\
&&\cup \left( \bigcup_{\{u,v\} \in E(G)} \left(\{C_{u,1}^v,...,C_{u,g}^v\} \cup \{C_{v,1}^u,...,C_{v,g}^u \}  \right)\right) \\
&&\cup  \left( \bigcup_{v \in V(G), 1\leq i < j \leq K} (\{D_{v,i,j,1},...,D_{v,i,j,g}\} \cup  \{D_{v,j,i,1},...,D_{v,j,i,g}\}) \right) \cup \{X\}
\end{eqnarray*} 

We now describe the $n = g + ge + g^2 \binom K 2 \leq g^3 + g^{2+2\epsilon}$ users' preference-lists. We start with the $g$ rankings specified in Table~\ref{fig:type1}. We continue with $ge$ more rankings, where for each edge $(u,v) \in E(G)$ we add $g$ more rankings, as detailed in Table~\ref{fig:type2}. Lastly, we add $g^2\binom{K} {2}$ more rankings, where for each triple $(v,i,j)$ where $v$ is a vertex of $G$ and $i\neq j \in [K]$ we add $g$ rankings, as detailed in Table~\ref{fig:type3}. (Observe, the tables detail the first few words in each list, then end with ``$\ldots$'' mark, which indicates that from that point on the remaining words may appear in any order.)

\begin{table}
\caption{Rankings used in the proof of Theorem~\ref{thm:HardnessOfApproximationOptimalCombinationOfSubsets}.}
\centering
\subfigure[First type.]{
\label{fig:type1}
\begin{tabular}{| c | c | c |}
\hline
$\ell_1$ & \ldots & $\ell_g$ \\
\hline
\hline
$A_1$ & $\ldots $& $A_1$ \\
\hline
$A_2$ & $\ldots $& $A_2$ \\
\hline 
\multicolumn{3}{|c|}{$\ldots$} \\
\hline
$A_K$ & $\ldots $& $A_K$ \\
\hline 
$B_1$ & $\ldots$ & $B_g$ \\
\hline
\multicolumn{3}{|c|}{$\ldots$} \\
\hline 
\end{tabular}
}
\quad
\subfigure[Second type.]{
\label{fig:type2}
\begin{tabular}{ | c | c | c |}
\hline
$\ell_{u,v,1}$ & \ldots & $\ell_{u,v,g}$ \\
\hline
\hline
$C_{u,1}^v$ & \ldots & $ C_{u,g}^v$ \\
\hline
$C_{v,1}^u$ & $\ldots $& $C_{v,g}^u$ \\
\hline
$X$ & $\ldots $& $X$ \\
\hline 
\multicolumn{3}{|c|}{$\ldots$} \\
\hline
\end{tabular}
}
\quad
\subfigure[Third type.]{
\label{fig:type3}
\begin{tabular}{ | c | c | c |}
\hline
$\ell_{v,i,j,1}$ & \ldots & $\ell_{v,i,j,g}$ \\
\hline
\hline
$D_{v,i,j,1} $ & \ldots & $ D_{v,i,j,g}$ \\
\hline
$D_{v,j,i,1} $ & $\ldots $& $ D_{v,j,i,g}$ \\
\hline
$X$ & $\ldots $& $X$ \\
\hline 
\multicolumn{3}{|c|}{$\ldots$} \\
\hline
\end{tabular}
}
\end{table}

Finally, we detail our rules. For every $i\in [K]$ and $u\in V(G)$ we have a rule which roughly corresponds to deciding that $u$ is a member of the independent set:
\[R_{u,i} = \{A_i\} \cup \bigcup_{\{v : ~ (u,v) \in E(G)\}} \{ C_{u,1}^v, C_{u,2}^v, \ldots, C_{u,g}^v\} \cup \bigcup_{j\in [K], j\neq i}\{D_{u,i,j,1}, \ldots, D_{u,i,j,g}\} \ .\]

Our analysis now follows from a series of observations.

\medskip

\noindent{\em Observation 1:} If we do not ban all of the passwords $A_1,..., A_K$ then $p_1 \geq g/n$.  Therefore, for every $i$, we must choose at least one of the rules $\{R_{u,i} \}$ to activate, or else we have that $p_1 \geq g/n$

\medskip

\noindent{\em Observation 2:} If we ban $C_{u,1}^v, \ldots, C_{u,g}^v$ and $C_{v,1}^u, \ldots, C_{v,g}^u$ then we must have $p_1 \geq g/n$. Therefore, for any $i\neq j$ it must not be the case that we ban $R_{u,i}$ and $R_{v,j}$ where $(u,v)\in E(G)$, or else we have that $p_1 \geq g/n$.

\medskip

\noindent{\em Observation 3:} If we ban $D_{v,i,j,1}, \ldots, D_{v,i,j,g}$, and $D_{v,j,i,1} , \ldots, D_{v,j,i,g}$ then $p_1\geq g/n$. Therefore, for any $i\neq j$ it must not be the case that we ban $R_{u,i}$ and $R_{u,j}$, or else we have that $p_1 \geq g/n$.

\medskip 

These observations lead us to the following conclusion. If $G$ contains an independent set $v_1,...,v_K$ of size $K$, then activating the rules $\{R_{v_1,1}, R_{v_2,2}, \ldots, R_{v_K,K}\}$ leads to a setting where each truncated ranking begins with a unique word, so $p_1 = 1/n$. In contrast, if $G$ does not have an independent set of size $K$, then $p_1 = g/n$. 
Since $n = O(g^3)$ we have an $\Omega(n^{1/3})$-hardness of approximation. \ap{Are we actually reducing from gap-IS here and not from the regular IS? It seems to me that $K$ could simply be the size of the independent set in a regular IS instance.} \os{The reduction works for any Max-Ind-Set instance. The approximation hardness follows from graphs where it hard to determine the size of the max-ind-set up to a factor of $\sqrt g$ -- whether $K = o(g^{1/2})$ or $K\geq g^{1-\epsilon}$. Observe, you need $K < g^{1/2}$ is you want to get $n^{1/3}$-hardness, otherwise $g^2K^2$ is the dominating factor for $n$ (also, in these hard instances $e = O(g^2)$.} \ap{Also, why do we need the $\epsilon$ in the inapproximability statement?}\os{Just a standard for apx-hardness -- it takes care of the $\Omega$-notation...} Observe also that the number of total words is $N = K+g+2eg + g^2K(K-1) +1= O(g^3) = O(n)$ so it is also hard to approximate the problem to a factor of $\Omega(N^{1/3})$.
\end{proofof}

\begin{remindertheorem}{\ref{thm:NrmPrbEfficientAlgAnyK}}
\thmNrmPrbEfficientAlgAnyK
\end{remindertheorem}

\medskip

\begin{proofof}{Theorem \ref{thm:NrmPrbEfficientAlgAnyK}}
Let $\mathcal{A}^*$ denote the optimal solution, denote its most $k$ popular passwords as $w_{i_1}, \ldots, w_{i_k}$, and denote also $P^*$ as the total probability mass of the words in $\mathcal{A}^*$ according to the initial distribution: $P^* = \sum_{w\in \mathcal A^*} \Pr[w]$. Therefore, $p(k,\mathcal{A}^*) = \sum_{j=1}^k \Pr[w_{i_j}] / P^*$. 

Clearly, all words $w_j$ s.t. $j > i_k$ belong to $\mathcal{A^*}$ -- otherwise, we could add such a word and decrease the probability of the top $k$ words. Similarly, all words $w_j$ s.t $j < i_1$ must not belong to $\mathcal A^*$, otherwise they would belong to the set of most popular $k$ words. We now claim that $w_{i_1}, \ldots, w_{i_k}$ are $k$ consecutive words.

Suppose that there was some word $w'$ between some $w_{i_j}$ and $w_{i_{j+1}}$. Then $\mathcal{A}^*$ clearly banned it, otherwise it would be one of the most popular $k$ words. We claim that the policy $\mathcal{A}'$ where we ban $w_{i_1}$ and allow $w'$ instead satisfies $p(k,\mathcal{A}') \leq p(k,\mathcal{A}^*)$.

We denote $p_1 = \Pr[w_{i_1}]$, $q = \sum_{j=2}^k \Pr[w_{i_j}]$ and $p' = \Pr[w']$, and we know $p_1 \geq p'$. Then $p(k,\mathcal{A}^*) = (p_1 + q) / P^*$, whereas 
\[ p(k,\mathcal{A}') =  \frac {p' + q}  {P^* - p_1 + p'} \ .\]
Our goal is to show $p(k,\mathcal{A}') \leq p(k,\mathcal{A}^*)$, which holds iff
\[ (p'+q) P^* \leq (p_1+q)(P^*- (p_1 - p'))  \]
By some algebraic manipulations, this holds iff \[  (p_1 - p') P^*  \geq (p_1 - p') (p_1 + q)  \] which clearly holds because $p_1 - p'$ is a non-negative quantity, and $p_1 + q = \sum_{j=1}^k \Pr[w_{i_j}] \leq \sum_{w\in \mathcal{A}^*} \Pr[w]$.

As for the running time of the algorithm, it is obvious that sorting requires $O(N\log N)$ time. Finding the minimum requires only $O(N)$ time: if we denote $a_ i =  \sum_{i \leq j \leq i+k} \Pr[w_j]$ and $b_i = \sum_{i \leq j} \Pr[w_j]$, then based on $a_i$ and $b_i$ it is easy to compute $a_{i+1}$ and $b_{i+1}$ in $O(1)$ time.
\end{proofof}

\begin{reminderclaim}{\ref{clm:FamilyOfSetsWithLargeUnions}}
Fix $m$ and $s$ such that $m\geq s$. There exists a domain $D$ of size $\Theta(s^2\log(m))$ and a family of $m$ sets, $F_1, F_2, \ldots, F_m\subseteq D$, such that each set in the family contains $\tfrac {|D|}{2s}$ elements, and for every $C\subseteq [m]$ of size $|C| \leq s$, we have that the size of the union $\left| \bigcup_{i\in C} F_i \right| \geq \tfrac {|D|}{2s} \tfrac{|C|} 4$. This domain can be constructed in randomized $\poly(s,m)$ time.
\end{reminderclaim}

\medskip

\begin{proofof}{Claim \ref{clm:FamilyOfSetsWithLargeUnions}} Given $m$ and $s$, we first pick a random function $\phi:[m] \to [2s]$. Fixing a subset $C\subseteq [m]$ of size $|C| \leq s$, we claim that $|\phi(C)| > |C|/2$ w.p. at least $1-(0.825)^{|C|}$. Indeed,
\begin{eqnarray*}
& \Pr\left[ |\phi(C)| \leq |C|/2 \right] & \leq \Pr\left[ \exists T\subseteq [2s] \textrm{ s.t. } |T| = |C|/2 \textrm{ and } \forall i \in C, \phi(i)\in T\right] \cr
&& \leq \binom{2s} {|C|/2} \Pr\left[\forall i \in C, \phi(i)\in T\right] \leq \left(\frac {4se}{|C|}\right)^{|C|/2} \left(\frac {|C|/2} {2s}\right)^{|C|} \cr
&& = e^{|C|/2} \left(\frac {|C|} {4s}\right)^{|C|/2} = \left(\sqrt{e/4}\right)^{|C|} < (0.825)^{|C|} \ .
\end{eqnarray*}
So assuming $|C| \geq 8$ we have that $C$ is mapped to at least $|C|/2$ distinct images by $\phi$ w.p.$>3/4$. Also,  if $|C|\leq 7$ then probability of even two elements getting mapped to the same image is at most $\binom{7} {2} \tfrac 1 {2s} < 0.25$ for $s> 42$.

We now construct $D$ by taking $d$ independently chosen such $\phi$-mappings, which we denote as $\phi_1, \phi_2, \ldots, \phi_d$, and so $D = [2s] \times [d]$. We construct the family $F_i = \{(\phi_1(i),1), (\phi_2(i),2), \ldots, (\phi_d(i),d)  \}$ for every $i \in [m]$. Clearly, for every $i$ it holds that $|F_i| = d = |D|/2s$. Supposed for the sake of contradiction that there exists some $C\subseteq [m]$ of size $\leq s$ such that $\left| \bigcup_{i\in C} F_i \right| \leq \tfrac {|C|} 4 |F_i|$. By construction, we have that 
\[\left| \bigcup_{i\in C} F_i \right| = \sum_{j=1}^d \left| \{ (\phi_j(C), j)  \}  \right| = \sum_{j=1}^d |\phi_j(C)|\]
so by the Markov inequality we have that at least $d/2$ functions where the cardinality of the image of $C$ is less than $|C|/2$. Let $X_{C,j}$ be the indicator random variable of $\phi_j$ mapping the set $C$ to no more than $|C|/2$ distinct elements, the Hoeffding bound gives that
\begin{eqnarray*}
& \Pr\left[ \exists C  \textrm{ of size $\leq s$ s.t. } \sum_j X_{C,j} > d/2 \right] &\leq \sum_{s'<s} \binom{m} {s'} \Pr[ \tfrac 1 d \sum_j X_{C,j} > 0.5] \leq m^{O(s)} e^{-d/10} 
\end{eqnarray*}
Setting $d = \Theta(s\log m)$ gives that w.p. $\geq 1/2$ no such $C$ exists.
\end{proofof}

\begin{remindertheorem}{\ref{thm:HardnessAddingSubsetsLargeK}}
\thmHardnessAddingSubsetsLargeK
\end{remindertheorem}

\medskip

\begin{proofof}{Theorem \ref{thm:HardnessAddingSubsetsLargeK}} 
Our reduction is from set cover. 

\noindent {\em Set Cover Instance:} Sets $S_1,\ldots,S_m$, Universe $U = \{1,\ldots,n\}$ and integer $k$.\\
\noindent {\em Question:} Is there a set cover of size $k-1$?

Now we define $W_1,\ldots,W_n$ to be $n$ disjoint sets of passwords \[W_i = \left\{w_{i,\ell}~\vline~1 \leq \ell \leq n^5m^5 \right\} \ . \] We also define special passwords $t_{j}$ ($j \leq m$) and $\tau_j$ ($j \leq k$)  which are not contained in any $W_i$. \\

We define the following positive password rules: 
\[ R_i = \{t_i\} \cup  \left\{\tau_{j} ~\vline~ 1\leq j \leq k\right\} + \bigcup_{j: j \in S_i} W_j \ .\]

We assign probabilities as follows:

\[ \Pr\left[w_{i,\ell} \right]  = \left(1-\frac{1}{n^3} \right) \frac{1}{m^5n^6} \ , \]
for each $i \leq n$ and $\ell \leq m^5n^5$. Observe that
\[ \Pr\left[\bigcup_{i \leq m} W_i \right]  = \left(1-\frac{1}{n^3} \right) \ , \] 
so that almost all of the probability mass is concentrated inside the sets $W_i$ and the probability mass is uniformly distributed. We also set
\[\Pr\left[\tau_j \right] = \frac{1-x}{n^3k} \ ,\]
and 
\[\Pr\left[t_j \right] = \frac{x}{n^3m} \ , \]
where $0 \leq x \leq 1$ will be defined later. First notice that 
\[\sum_{j \leq k} \tau_j + \sum_{j \leq m} t_j = k\left(\frac{1-x}{n^3k}\right) + m\left(\frac{x}{n^3m}\right) = \frac{1}{n^3}  \ ,\]
so our probability distribution is well defined. 
Suppose that there is a set cover $C \subseteq [m]$ s.t. $\left|C \right| \leq k-1 \wedge \bigcup_{i \in C} S_i = U$, and consider the solution $\mathcal{A}_C$. We cover all $W_i$'s and use at most $k-1$ $t$'s. Hence, 
\[p\left(k,\mathcal{A}_C\right) \leq \left(\left(k-1\right)\Pr[t] + \Pr[\tau]\right)\left(\frac{n^3}{n^3-1} \right) \ .\]
Suppose that there is no set cover of size $k$. For every set of $k$ or more rules $S$ we have at least $k$ $t$'s in our solution so 
\[p\left(k,\mathcal{A}_S\right) \geq k\Pr[t]  \ . \]
For every set of rules $S$ that does not cover all the $W_i$'s we have at most $ \left(1-\frac{1}{n}\right)\left(1-\frac{1}{n^3} \right) $-fraction of the total probability mass so
\[p\left(k,\mathcal{A}_S\right) \geq \frac{\left((k-1)\Pr[\tau]+\Pr[t]\right)}{\left(1-\frac{1}{n}\right)\left(1-\frac{1}{n^3} 
\right) } \ . \] 

It suffices to select $x$ s.t. 
\[\left(\left(k-1\right)\Pr[t] + \Pr[\tau]\right)\left(\frac{n^3}{n^3-1} \right) < \min\left\{\frac{\left((k-1)\Pr[\tau]+\Pr[t]\right)}{\left(1-\frac{1}{n}\right)\left(1-\frac{1}{n^3} 
\right) }  , k\Pr[t] \right\} \ , \]
or ---after some algebraic manipulation --- equivalently,
\[ a= \frac{\left(\frac{n^3}{n^3-1} \right)}{\left(1-\frac{1}{n^3-1}\right)}\Pr[\tau] <  \Pr[t] < b =\Pr[\tau] \frac{(k-2)+\frac{1}{n-1}}{(k-2)-\frac{1}{n-1}} \ . \]
Observe that $a \leq \Pr[\tau] \leq b$ so it suffices to set $x$ s.t. $\Pr[t] = \frac{a+b}{2}$. We can solve for x to get
\[x = \frac{m\left(-3+2n+2n^3-2n^4+k\left(n-1\right)^2\left(1+n+n^2\right)\right)}{m\left(-3+2n+2n^3-2n^4\right) + k^2\left(2-2n-n^3+n^4 \right) + k\left(-2+4n+n^3-2n^4+m\left(n-1\right)^2\left(1+n+n^2\right) \right)}  \ .\]
\end{proofof}

\begin{reminderclaim}{\ref{claim:SamplingBadEvent}}
$ \Pr\left[\exists i, BAD_i\right] \leq \delta \  . $ 
\end{reminderclaim} 

\medskip

\begin{proofof}{Claim \ref{claim:SamplingBadEvent}} 
By the union bound it suffices to show that 
\[ \Pr\left[BAD_i\right] \leq \frac{\delta}{m} \ . \]
Our first step is to divide the passwords $w \in \PasswordSpace$ into buckets $B_j$ based on their probability. For $j > 0$ we define
  \[ B_j = \left\{ w~\vline \frac{\epsilon}{2^{j}} \leq \Pr\left[w ~\vline ~\mathcal{A}_{S_i} \right] \leq \frac{\epsilon}{2^{j-1}} \right\} \ , \]
and for $j=0$ we set \[ B_0 = \left\{ w~\vline ~\epsilon \leq \Pr\left[w ~\vline ~\mathcal{A}_{S_i} \right]\right\}  \ . \]
Observe that 
\[ \PasswordSpace = \bigcup_{j=0}^\infty B_j \ . \]

Let $w \in B_j$ be given ($j > 0$) then by the Chernoff Bounds:
\[\Pr\left[s_w > s \Pr\left[w ~\vline ~\mathcal{A}_{S_i} \right] + s\epsilon/2  \right] \leq \exp\left(-2^{j-1} \log \left(\frac{4m}{\delta \epsilon} \right) \right) \leq \frac{4^{-2^{j-1}} \delta \epsilon}{m} \ . \]
Notice that the bucket $B_j$ contains at most $\left|B_j\right| = 2^{j}/\epsilon$ passwords. 
\[\Pr\left[\exists w \in B_j, s_w > s \Pr\left[w ~\vline ~\mathcal{A}_{S_i} \right] + s\epsilon/2  \right]  \leq \frac{4^{-2^{j-1}} \delta \epsilon\left|B_j\right|}{m} \leq  \frac{ \delta }{2^{j+1}m} \ . \]
Now if we union bound across all $j>0$ we get  
 \[\Pr\left[\exists w \in \bigcup_{j=1}^\infty B_j, s_w > s \Pr\left[w ~\vline ~\mathcal{A}_{S_i} \right] + s\epsilon/2  \right] \leq \sum_{j=1}^\infty \frac{ \delta }{2^{j+1}m} = \frac{\delta}{2m} \ . \]
Finally, we consider the passwords in $B_0$. By Chernoff Bounds for each $w \in B_0$ we have
\[\Pr\left[\left| s_w - s \Pr\left[w ~\vline ~\mathcal{A}_{S_i} \right] \right| > s\epsilon/2  \right] \leq \frac{\delta \epsilon}{2m} \ , \]
by applying the union bound $\left|B_0\right| \leq 1/\epsilon$ we get
\[\Pr\left[\exists w \in B_0 \left| s_w - s \Pr\left[w ~\vline ~\mathcal{A}_{S_i} \right] \right| > s\epsilon/2  \right] \leq \frac{\delta }{2m} \ . \]
Combining our inequalities we obtain the desired result: 
\[\Pr\left[BAD_i \right] \leq \Pr\left[\exists w \in \bigcup_{j=0}^\infty B_j, s_w > s \Pr\left[w ~\vline ~\mathcal{A}_{S_i} \right] + s\epsilon/2  \right] \leq \frac{\delta}{m} \ . \]
\end{proofof}

\section{Impossibility of constant-factor universal approximation} \label{subsec:ImposibilityOfConstantFactorUniversalAppx}
In this section we consider the following goal: given a constant $c$ find a password composition policy $\mathcal{A}$ such that 
\[ p\left(k,\mathcal{A} \right) \leq c \cdot p\left(k,\mathcal{A}'\right) \  , \]
for any other policy $\mathcal{A} '$ and {\em every} value of $k \leq N$. Such a policy --- if it exists --- would provide a nearly optimal defense against both online attacks and dictionary attacks simultaneously \cite{seeley1989password}. Unfortunately, Theorem \ref{thm:ImpossibilityOfConstantFactorUniversalAppx} rules out the possibility of a constant universal approximation in the rankings model. Our impossibility result holds even in the singleton rules setting. We show that it is possible to construct a distribution $\mathcal{D}$ over rankings for which no universal approximation exists.

We construct our distribution $\mathcal{D}$ (algorithm \ref{alg:DistributionImpossible}) over rankings by merging two distributions $\mathcal{D}_1$ and $\mathcal{D}_2$ over preference lists. 

{\em Intuition:} Passwords sampled from $D_2$ are highly secure, but passwords sampled from $D_1$ are highly insecure. To make improve the security of $D_1$ it is necessary to ban all passwords in $W$, but this reduces the security of $D_2$ significantly. 

We make two claims (1) We must ban {\em all} but a small subset of passwords if we want to even {\em approximately} optimize  $p\left(1,\mathcal{A} \right)$. (2) We must keep a larger subset of passwords to even {\em approximately} optimize $p\left(k,\mathcal{A}\right)$ for large values of $k$.

\begin{theorem} \label{thm:ImpossibilityOfConstantFactorUniversalAppx}
For all constants $c>0$ there exists distribution $\mathcal{D}$ over rankings such that $\forall \mathcal{A} \subseteq \PasswordSpace, \exists \mathcal{A}', k \in \mathbb{N},$ such that 
 \[ p\left(k,\mathcal{A} \right) > c \cdot p\left(k,\mathcal{A}'\right) \ . \]
\end{theorem}

\begin{proof}(sketch) 
Let $\PasswordSpace = W \cup X$ where $W= \bigcup_{i=1}^{r} W_i$ --- $W_i = \{w_{i,1},\ldots,w_{i,t}\}$ --- and $X= \{x_1,\ldots,x_{L}\}$ are two disjoint sets of passwords, where the parameters are set as follows
$q = \frac{1}{2c}$, $t=L = \log N$ and $r = \frac{N-L}{t}$.
 Our distribution over preference lists is given by algorithm \ref{alg:DistributionImpossible}.

\begin{algorithm}
\caption{Sample $\mathcal{D}$}
\begin{algorithmic}
\State {\bf Input:}
\State Parameters $L, r, q, t$
\State Random Number  $u \in [0,1]$. 
\State Random Permutation $\pi_i$ of $W_i$ for each $i \in \{1,...,r\}$
\State Random Permutation $pi_X$ over $X$
\State Random Permutation $\pi_\PasswordSpace$ of $\PasswordSpace$ 
\State {\bf Initialize: } $\ell \gets$ empty ranking
\If{$u \leq q$ } 
\Comment{Select from $\mathcal{D}_1$}
 \For{$i = 1 \to r$}
 \State $\ell \gets \langle\ell,\pi_{10^r}\rangle$
 \Comment{Append random permutation of $W_i$}
 \EndFor
 \State $\ell \gets \langle \ell , \pi_X \rangle$
 \Comment{Append random permutation of $X$}
\Else
\Comment{Select from $\mathcal{D}_2$}
\State $\ell \gets \pi_\PasswordSpace$
\EndIf
\Return $\ell$
\end{algorithmic}
\label{alg:DistributionImpossible}
\end{algorithm}

There are two cases to consider:

{\em Case 1:} $\exists x \in W-\mathcal{A}$ then it is easy to see that 
\[ p\left(1,\mathcal{A}\right) \geq \frac{q}{t} = \frac{2c}{t} = \frac{2c}{L} \geq 2c \times p\left(1,X\right)   \ . \]

{\em Case 2:} Suppose that $\forall x \in W$ we have $x \notin \mathcal{A}$ and consider $k=L$ with the solution $\PasswordSpace$ --- don't ban any passwords. For the solution $\PasswordSpace$ we have \[ p_i = \frac{q}{t}+\frac{1-q}{\left|X\right|+\left|W\right|} \ ,\] for $i \leq t$ (e.g., for the $t$ the passwords in $W_1$), and \[p_i = \frac{1-q}{\left|X\right|+\left|W\right|} \ , \]
for $i >  t$.

\begin{eqnarray*}
c \times p\left(k,\PasswordSpace\right) &=&c\sum_{i=1}^t \left(\frac{q}{t} + \frac{1-q}{|X|+|W|} \right) + c\sum_{i=t+1}^k \frac{1-q}{|X|+|W|}  \\
&=& c\left(q+\left(1-q\right)\frac{L}{L+10^r} \right) \\
&=& \frac{1}{2} + \left(c-\frac{1}{2}\right)\frac{L}{L+10^r} \\
&<& 1 =  p\left(k,\mathcal{A}\right) \ .
\end{eqnarray*}
\end{proof}

\section{Experiment Rules} \label{subsec:experimentsetup}

We selected rules based on common types of rules used in constructing password composition policies, e.g., the policies recommended by NIST~\cite{burr_electronic_2006}. The rules we selected are shown in Table~\ref{tab:ExperimentRules}. Positive and negative forms of each rule are shown. In the positive rules setting, a password is allowed if it matches any positive rule. In the negative rules setting, a password is banned if it matches any negative rule.

\begin{table}[htb]
\centering
 \tbl{Rules Used in Sampling Experiments}{
\begin{tabular}{| p{1.6in} | p{1.6in} | p{1.2in} |}
\hline
{\textbf{Positive Rule}} & {\textbf{Negative Rule}} & {\textbf{Details}} \\
\hline
8 characters or more & Less than 8 characters & \multirow{9}{*}{Length rules} \\
\cline{1-2}
9 characters or more & Less than 9 characters & \\
\cline{1-2}
10 characters or more & Less than 10 characters & \\
\cline{1-2}
11 characters or more & Less than 11 characters & \\
\cline{1-2}
12 characters or more & Less than 12 characters & \\
\cline{1-2}
13 characters or more & Less than 13 characters & \\
\cline{1-2}
14 characters or more & Less than 14 characters & \\
\cline{1-2}
15 characters or more & Less than 15 characters & \\
\cline{1-2}
16 characters or more & Less than 16 characters & \\
\hline
1 digit or more & Less than 1 digit & \multirow{8}{*}{Character class rules}\\
\cline{1-2}
1 symbol or more & Less than 1 symbol & \\
\cline{1-2}
1 lowercase or more & Less than 1 lowercase & \\
\cline{1-2}
1 uppercase or more & Less than 1 uppercase & \\
\cline{1-2}
2 digits or more & Less than 2 digits & \\
\cline{1-2}
2 symbols or more & Less than 2 symbols & \\
\cline{1-2}
2 lowercase or more & Less than 2 lowercase & \\
\cline{1-2}
2 uppercase or more & Less than 2 uppercase & \\
\hline
In a dictionary & Not in a dictionary & \multirow{2}{*}{Dictionary checks} \\
\cline{1-2} 
Contains a dictionary word & Does not contain a dictionary word &  \\
\hline 
8 characters or more AND 1 uppercase or more & Less than 8 characters OR less than 1 uppercase & \multirow{2}{*}{Combination Rules} \\
\cline{1-2} 
8 characters or more AND 1 uppercase or more AND 1 digit or more & Less than 8 characters OR less than 1 uppercase OR less than 1 digit & \\
\hline
\end{tabular}}
\label{tab:ExperimentRules}
\end{table}

The dictionary check used the cracking dictionary from \url{openwall.com}. This dictionary is used by one of the most well-known password crackers, John the Ripper~\cite{JTR}. Since this dictionary contains all alphabetic strings up to size 3, it was pruned to only include entries of 4 characters or more for the ``contains a dictionary word'' dictionary check.

Notice that for some groups of rules, e.g., length rules, digit rules, etc., the subsets defined by these rules are subsets or supersets of each other. For example, if the positive rule ``8 characters or more'' is in a policy, adding the ``10 characters or more'' rule yields the same policy. We did this to prevent the selection of overly complex policies, e.g., ``8 characters'' OR ``11 characters'' OR ``12 characters'' OR ``14 characters.'' However, we also selected a couple of ``combination rules'' to make policies more interesting.

\end{document}